%% file: main.tex
\documentclass[
  paper=23.9cm:16.9cm,
  pagesize,
  twoside,
  12pt,
  chapterprefix,
  headsepline=on,
  footinclude=off,
  DIV=18,
  BCOR=7mm,
  bibliography=totoc,
  numbers=noenddot,
  open=right,
]{scrreprt}
\usepackage[%
  left=20mm,
  right=20mm,
  top=25mm,
  bottom=25mm,
]{geometry}

\input{TemplateFiles/Packages.tex}

\usepackage{TemplateFiles/AddMyPapers}

\input{Packages.tex}

\input{Config.tex}

\newcommand{\ThesisTypePhD}{phd}
\newcommand{\ThesisTypeLic}{licentiate of engineering}
\ifx \TheThesisType\ThesisTypeLic
    
\else
    
    \ifx \TheThesisType\ThesisTypePhD \else
        \PackageError {ThesisTemplate} {Incorrect value of ``ThesisType``, choose either ``phd`` or ``lic``} {The variable ``ThesisType`` decides on the type of thesis to produce, choose ``phd`` to generate a PhD thesis or ``lic`` to generate a licentiate thesis.}
    \fi
\fi
\ifdefined \TheCoverPageSubtitle \if \TheCoverPageSubtitle\empty\else
    \ifdefined \TheSubtitle \if \TheSubtitle\empty\else
        \def\ConsistentSubtitle{true}
    \fi\fi
    \ifdefined \ConsistentSubtitle \else
        \PackageError {ThesisTemplate} {Definition of ``TheCoverPageSubtitle`` and ``TheSubtitle`` is inconsistent} {The variable ``TheCoverPageSubtitle`` can not be non-empty when the variable ``TheSubtitle`` is empty or undefined.}
    \fi
\fi\fi

\usepackage{hyperref}
\hypersetup{pdfcreator={LaTeX},	
	pdfkeywords={}, 
	pdfpagemode=UseOutlines, 
	pdfdisplaydoctitle=true, 
	pdflang=en, 
	colorlinks=false,	
	linkcolor=LinkColor, 
	citecolor=LinkColor, 
	filecolor=LinkColor, 
	urlcolor=LinkColor,
	menucolor=LinkColor, 
  bookmarks=true,
  bookmarksnumbered=true}

\hypersetup{
	pdftitle={Thesis for \TheThesisType}, 
	pdfauthor={\TheAuthor}, 
	pdfsubject={\TheTitle}
}
\usepackage{TemplateFiles/ChalmersThesis}

\begin{document}

\newpage
\thispagestyle{empty}

\begin{center}
\textsc{thesis for the degree of \TheThesisType}
\vspace{0.5cm}
\hrule
\vspace{3 cm}
{\Large {\TitlePageTitle}\par}
\vspace{1 cm}
{\large\textsc{\TitlePageName}\par} 
\vspace{2 cm}
\includegraphics[width=5cm]{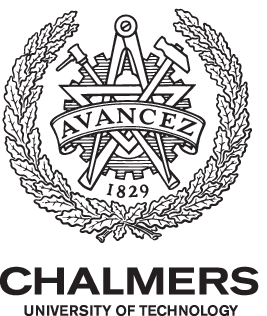}
\vfill
\vspace{0.5 cm}
Communication Systems Group\\
Department of Electrical Engineering \\
Chalmers University of Technology \\
Gothenburg, Sweden, \TheYear
\end{center}

\newpage

\thispagestyle{empty}

\noindent\textbf{\ISSNPageTitle}\\
\ifdefined \ISSNPageSubitle \if \ISSNPageSubitle\empty\else
    \noindent\footnotesize\textit{\ISSNPageSubitle}\\
\fi\fi\\
{\textsc{\ISSNPageName}}
\vspace{4.0 cm}\\
Copyright \copyright~\TheYear~\textsc{\ISSNPageName}\\ All rights reserved.\\
\vspace{1.5 cm}\\
ISSN \TheISSN\\
\vspace{1 cm}%
This thesis has been prepared using \LaTeX.\\
Communication Systems Group\\
Department of \TheDepartment \\
Chalmers University of Technology \\
SE-412 96 Gothenburg, Sweden \\
Phone: +46 (0)31 772 1000 \\
\texttt{www.chalmers.se} \\
\vfill
\noindent{Printed by \ThePrinter \\
Gothenburg, Sweden, \TheMonth~\TheYear}

\newpage


\thispagestyle{empty}
\centerline{ }
\vspace{1cm}

\rightline{{\emph{\TheDedication}}}

\thispagestyle{empty}

\EmptyPage 



\hypersetup{pageanchor=true}
\pagenumbering{roman}
\setcounter{page}{1}

\setlength{\parindent}{1em}

\section*{Abstract}
\addcontentsline{toc}{chapter}{Abstract}
\input{YourThesis/Abstract.tex}

\bigskip 
\noindent\textbf{Keywords:} \TheKeyword

\EmptyPageWithNumber 


\setlength{\parindent}{0pt}


\section*{List of Publications}
\addcontentsline{toc}{chapter}{List of Papers}
\includelistofpapers

\setlength{\parindent}{1em}

\EmptyPageWithNumber 


\section*{Acknowledgments}
\addcontentsline{toc}{chapter}{Acknowledgements}
\input{YourThesis/Acknowledgement.tex}

\newpage

\section*{Financial Support}
\input{FinancialSupport.tex}

\newpage

\section*{Acronyms}
\vspace{1cm}
\addcontentsline{toc}{chapter}{Acronyms}
\input{Acronyms.tex}



\tableofcontents
\clearpage

\EmptyPage

\setcounter{page}{1}
\renewcommand{\thepage}{\arabic{page}}

\setcounter{secnumdepth}{3}

\pagestyle{scrheadings}
\automark[section]{chapter}

\begin{refsection} 
\part{Overview}
\input{Content.tex}
\printbibliography[title=References]
\end{refsection}

\part{Papers}
\setcounter{secnumdepth}{2}

\initpapers
\renewcommand\bibname{References}

\automark{chapter}
\renewcommand\sectionmark[1]{\markright{\MakeMarkcase {\thesection\hskip .5em\relax#1}}}
\rohead{\rightmark}
\lehead{Paper \Alph{paper}}
\rofoot{\pagemark}
\lefoot{\pagemark}





\end{document}

%% file: TemplateFiles/Packages.tex
\usepackage{wrapfig} 
\usepackage{tikz} 
\usepackage{lmodern}
\usepackage[T1]{fontenc}
\usepackage[utf8]{inputenc}
\usepackage[english]{babel}	
\usepackage[babel,english=american]{csquotes}
\usepackage{scrlayer-scrpage}
\usepackage{scrhack}
\usepackage[hang,small]{caption} 
\usepackage[]{datetime}
\usepackage{amsthm,mathtools}
\usepackage[toc,page]{appendix}
\usepackage[Sonny]{fncychap}
\usepackage{bm}
\usepackage{array}
\usepackage{longtable}
\usepackage{booktabs}	
\usepackage{color}	
\usepackage{graphicx}
\usepackage{xspace} 
\usepackage{letltxmacro}
\usepackage{etoolbox}

\definecolor{LinkColor}{rgb}{0,0,0.5}

\usepackage[natbib=true,
			style=ieee,
            mincitenames=1,
            maxcitenames=2,
            url=false,
            eprint=false,
            doi=false]{biblatex}
\usepackage{silence}

\defbibheading{secbib}[\bibname]{%
  \section*{#1}%
  \addcontentsline{toc}{section}{\bibname}
  \markboth{#1}{#1}
  }

%% file: Packages.tex
\usepackage{graphicx,mathtools,amsmath,amssymb,amsfonts, bm, upgreek,acronym,enumerate,setspace,calc,epsf,tabularx,color,geometry,xspace}%

\usepackage{psfrag}

\usepackage{algorithm}
\usepackage[noend]{algpseudocode}
\usepackage[export]{adjustbox}

\usepackage{glossaries}
\input{YourThesis/glossary.tex}

\usepackage{graphicx}
\usepackage{bm}
\usepackage{enumerate}
\usepackage{float}
\usepackage{multicol}
\usepackage{color}
\usepackage{url}

\usepackage{amsmath}
\usepackage{amsthm}
\usepackage{amssymb}

\usepackage{epstopdf}

\usepackage{array}

\DeclareMathOperator*{\argmax}{argmax}
\usepackage{chngcntr}

\newtheorem{lem}{Lemma}
\counterwithin*{lem}{paper} 
\counterwithin*{corollary}{paper} 
\usepackage[toc,page]{appendix}

\acrodef{BS}{base station}
\acrodef{RA}{receive antenna}
\acrodef{PA}{predictor antenna}
\acrodef{IID}{independent and identically distributed}
\acrodef{CDF}{cumulative distribution function}
\acrodef{PDF}{probability density function}
\acrodef{cu}{channel use}
\acrodef{ACK}{acknowledgement}
\acrodef{NACK}{negative acknowledgement}
\acrodef{SNR}{signal-to-noise ratio}
\acrodef{HARQ}{hybrid automatic repeat request}
\acrodef{CSIT}{channel state information at the transmitter side}
\acrodef{INR}{incremental redundancy}
\acrodef{npcu}{nats per channel use}

%% file: YourThesis/glossary.tex
\newglossaryentry{5g}
{
  name={5G},
  description={fifth generation},
  first={\glsentrydesc{5g} (\glsentrytext{5g})}
}

\newglossaryentry{2g}
{
  name={2G},
  description={second generation},
  first={\glsentrydesc{2g} (\glsentrytext{2g})}
}

\newglossaryentry{3g}
{
  name={3G},
  description={third generation},
  first={\glsentrydesc{3g} (\glsentrytext{3g})}
}

\newglossaryentry{4g}
{
  name={4G},
  description={fourth generation},
  first={\glsentrydesc{4g} (\glsentrytext{4g})}
}

\newglossaryentry{lte}
{
  name={LTE},
  description={Long-Term Evolution},
  first={\glsentrydesc{lte} (\glsentrytext{lte})}
}

\newglossaryentry{iot}
{
  name={IoT},
  description={Internet of Things},
  first={\glsentrydesc{iot} (\glsentrytext{iot})}
}

\newglossaryentry{3gpp}
{
  name={3GPP},
  description={3rd generation partnership project},
  first={\glsentrydesc{3gpp} (\glsentrytext{3gpp})}
}

\newglossaryentry{csi}
{
  name={CSI},
  description={channel state information},
  first={\glsentrydesc{csi} (\glsentrytext{csi})}
}

\newglossaryentry{csit}
{
  name={CSIT},
  description={channel state information at the transmitter},
  first={\glsentrydesc{csit} (\glsentrytext{csit})}
}

\newglossaryentry{nr}
{
  name={NR},
  description={New Radio},
  first={\glsentrydesc{nr} (\glsentrytext{nr})}
}

\newglossaryentry{imt-2020}
{
  name={IMT-2020},
  description={International Mobile Telecommunications-2020},
  first={\glsentrydesc{imt-2020} (\glsentrytext{imt-2020})}
}

\newglossaryentry{itu}
{
  name={ITU},
  description={International Telecommunication Union},
  first={\glsentrydesc{itu} (\glsentrytext{itu})}
}

\newglossaryentry{emb}
{
  name={EMBB},
  description={Enhanced mobile broadband},
  first={\glsentrydesc{emb} (\glsentrytext{emb})}
}

\newglossaryentry{mtc}
{
  name={MTC},
  description={Massive machine-type communications},
  first={\glsentrydesc{mtc} (\glsentrytext{mtc})}
}

\newglossaryentry{urllc}
{
  name={URLLC},
  description={Ultra-reliable low-latency communications},
  first={\glsentrydesc{urllc} (\glsentrytext{urllc})}
}

\newglossaryentry{bs}
{
  name={BS},
  description={base station},
  first={\glsentrydesc{bs} (\glsentrytext{bs})}
}

\newglossaryentry{v2x}
{
  name={V2X},
  description={vehicle-to-everything},
  first={\glsentrydesc{v2x} (\glsentrytext{v2x})}
}

\newglossaryentry{pa}
{
  name={PA},
  description={predictor antenna},
  first={\glsentrydesc{pa} (\glsentrytext{pa})}
}

\newglossaryentry{pas}
{
  name={PAS},
  description={predictor antenna system},
  first={\glsentrydesc{pas} (\glsentrytext{pas})}
}

\newglossaryentry{ra}
{
  name={RA},
  description={receive antenna},
  first={\glsentrydesc{ra} (\glsentrytext{ra})}
}

\newglossaryentry{mrn}
{
  name={MRN},
  description={moving relay node},
  first={\glsentrydesc{mrn} (\glsentrytext{mrn})}
}

\newglossaryentry{nmse}
{
  name={NMSE},
  description={normalized mean squared error},
  first={\glsentrydesc{nmse} (\glsentrytext{nmse})}
}
\newglossaryentry{rf}
{
  name={RF},
  description={radio-frequency},
  first={\glsentrydesc{rf} (\glsentrytext{rf})}
}

\newglossaryentry{iab}
{
  name={IAB},
  description={integrated access and backhaul},
  first={\glsentrydesc{iab} (\glsentrytext{iab})}
}

\newglossaryentry{gps}
{
  name={GPS},
  description={global positioning system},
  first={\glsentrydesc{gps} (\glsentrytext{gps})}
}

\newglossaryentry{ack}
{
  name={ACK},
  description={acknowledgment},
  first={\glsentrydesc{ack} (\glsentrytext{ack})}
}

\newglossaryentry{nack}
{
  name={NACK},
  description={negative acknowledgment},
  first={\glsentrydesc{nack} (\glsentrytext{nack})}
}

\newglossaryentry{mimo}
{
  name={MIMO},
  description={multiple-input multiple-output},
  first={\glsentrydesc{mimo} (\glsentrytext{mimo})}
}
\newglossaryentry{miso}
{
  name={MISO},
  description={multiple-input single-output},
  first={\glsentrydesc{miso} (\glsentrytext{miso})}
}

\newglossaryentry{npcu}
{
  name={npcu},
  description={nat per channel use},
  first={\glsentrydesc{npcu} (\glsentrytext{npcu})}
}

\newglossaryentry{siso}
{
  name={SISO},
  description={single-input single-output},
  first={\glsentrydesc{siso} (\glsentrytext{siso})}
}

\newglossaryentry{comp}
{
  name={CoMP},
  description={coordinated multipoint},
  first={\glsentrydesc{comp} (\glsentrytext{comp})}
}

\newglossaryentry{jt}
{
  name={JT},
  description={joint transmission},
  first={\glsentrydesc{jt} (\glsentrytext{jt})}
}

\newglossaryentry{ul}
{
  name={UL},
  description={uplink},
  first={\glsentrydesc{ul} (\glsentrytext{ul})}
}

\newglossaryentry{dl}
{
  name={DL},
  description={downlink},
  first={\glsentrydesc{dl} (\glsentrytext{dl})}
}

\newglossaryentry{tdd}
{
  name={TDD},
  description={time division duplex},
  first={\glsentrydesc{tdd} (\glsentrytext{tdd})}
}

\newglossaryentry{fdd}
{
  name={FDD},
  description={frequency division duplex},
  first={\glsentrydesc{fdd} (\glsentrytext{fdd})}
}

\newglossaryentry{bsc}
{
  name={BSC},
  description={binary-symmetric channel},
  first={\glsentrydesc{bsc} (\glsentrytext{bsc})}
}

\newglossaryentry{wp1}
{
  name={w.p.1.},
  description={with probability one},
  first={\glsentrydesc{wp1} (\glsentrytext{wp1})}
}

\newglossaryentry{arq}
{
  name={ARQ},
  description={automatic repeat request},
  first={\glsentrydesc{arq} (\glsentrytext{arq})}
}

\newglossaryentry{harq}
{
  name={HARQ},
  description={hybrid automatic repeat request},
  first={\glsentrydesc{harq} (\glsentrytext{harq})}
}

\newglossaryentry{dmc}
{
  name={DMC},
  description={discrete memoryless channel},
  first={\glsentrydesc{dmc} (\glsentrytext{dmc})}
}

\newglossaryentry{bpsk}
{
  name={BPSK},
  description={binary phase-shift keying},
  first={\glsentrydesc{bpsk} (\glsentrytext{bpsk})}
}

\newglossaryentry{wss-us}
{
  name={WSS-US},
  description={wide-sense stationary with uncorrelated scatterers},
  first={\glsentrydesc{wss-us} (\glsentrytext{wss-us})}
}

\newglossaryentry{awgn}
{
  name={AWGN},
  description={additive white Gaussian channel},
  first={\glsentrydesc{awgn} (\glsentrytext{awgn})}
}

\newglossaryentry{biawgn}
{
  name={BI-AWGN},
  description={binary additive white Gaussian channel},
  first={\glsentrydesc{biawgn} (\glsentrytext{biawgn})}
}

\newglossaryentry{LED}
{
  name={LED},
  description={light emitting diode},
  first={\glsentrydesc{LED} (\glsentrytext{LED})},
  plural={LEDs},
  descriptionplural={light emitting diodes},
  firstplural={\glsentrydescplural{LED} (\glsentryplural{LED})}
}

\newglossaryentry{rtd}
{
  name={RTD},
  description={repetition time diversity},
  first={\glsentrydesc{rtd} (\glsentrytext{rtd})}
}

\newglossaryentry{inr}
{
  name={INR},
  description={incremental redundancy},
  first={\glsentrydesc{inr} (\glsentrytext{inr})}
}

\newglossaryentry{ofdm}
{
  name={OFDM},
  description={orthogonal frequency-division multiplexing},
  first={\glsentrydesc{ofdm} (\glsentrytext{ofdm})}
}
\newglossaryentry{ue}
{
  name={UE},
  description={user equipment},
  first={\glsentrydesc{ue} (\glsentrytext{ue})},
  plural={UE's},
  descriptionplural={user equipments},
  firstplural={\glsentrydescplural{ue} (\glsentryplural{ue})}
}

\newglossaryentry{rb}
{
  name={RB},
  description={resource block},
  first={\glsentrydesc{rb} (\glsentrytext{rb})},
  plural={RBs},
  descriptionplural={resource blocks},
  firstplural={\glsentrydescplural{rb} (\glsentryplural{rb})}
}
\newglossaryentry{cu}
{
  name={CU},
  description={channel use},
  first={\glsentrydesc{cu} (\glsentrytext{cu})},
  plural={CUs},
  descriptionplural={channel uses},
  firstplural={\glsentrydescplural{cu} (\glsentryplural{cu})}
}
\newglossaryentry{tti}
{
  name={TTI},
  description={transmission time interval},
  first={\glsentrydesc{tti} (\glsentrytext{tti})},
  plural={TTI's},
  descriptionplural={transmission time intervals},
  firstplural={\glsentrydescplural{tti} (\glsentryplural{tti})}
}
\newglossaryentry{iid}
{
  name={i.i.d.},
  description={identical and independently distributed},
  first={\glsentrydesc{iid} (\glsentrytext{iid})}
}
\newglossaryentry{psd}
{
  name={PSD},
  description={power spectral density},
  first={\glsentrydesc{psd} (\glsentrytext{psd})}
}

\newglossaryentry{bler}
{
  name={BLER},
  description={block error probability},
  first={\glsentrydesc{bler} (\glsentrytext{bler})}
}

\newglossaryentry{dvbt}
{
  name={DVB-T},
  description={terrestrial digital video broadcasting},
  first={\glsentrydesc{dvbt} (\glsentrytext{dvbt})}
}
\newglossaryentry{pat}
{
  name={PAT},
  description={pilot-assisted transmission},
  first={\glsentrydesc{pat} (\glsentrytext{pat})}
}
\newglossaryentry{ml}
{
  name={ML},
  description={maximum likelihood},
  first={\glsentrydesc{ml} (\glsentrytext{ml})}
}
\newglossaryentry{dt}
{
  name={DT},
  description={dependence-testing},
  first={\glsentrydesc{dt} (\glsentrytext{dt})}
}
\newglossaryentry{ustm}
{
  name={USTM},
  description={unitary space-time modulation},
  first={\glsentrydesc{ustm} (\glsentrytext{ustm})}
}

\newglossaryentry{snn}
{
  name={SNN},
  description={scaled nearest-neighbor},
  first={\glsentrydesc{snn} (\glsentrytext{snn})}
}

\newglossaryentry{rcus}
{
  name={RCUs},
  description={random-coding union bound with parameter $s$},
  first={\glsentrydesc{rcus} (\glsentrytext{rcus})}
}
\newglossaryentry{osd}
{
  name={OSD},
  description={ordered statistics decoding},
  first={\glsentrydesc{osd} (\glsentrytext{osd})}
}
\newglossaryentry{llr}
{
  name={LLR},
  description={log-likelihood ratio},
  first={\glsentrydesc{llr} (\glsentrytext{llr})}
   plural={LLR's},
  descriptionplural={log-likelihood ratios},
  firstplural={\glsentrydescplural{llr} (\glsentryplural{llr})}
}
\newglossaryentry{qpsk}
{
  name={QPSK},
  description={quaternary phase-shift keying},
  first={\glsentrydesc{qpsk} (\glsentrytext{qpsk})}
}
\newglossaryentry{nlos}
{
  name={NLOS},
  description={non-line-of-sight},
  first={\glsentrydesc{nlos} (\glsentrytext{nlos})}
}
\newglossaryentry{los}
{
  name={LOS},
  description={line-of-sight},
  first={\glsentrydesc{los} (\glsentrytext{los})}
}
\newglossaryentry{cgf}
{
  name={CGF},
  description={cumulant-generating function},
  first={\glsentrydesc{cgf} (\glsentrytext{cgf})}
   plural={CGF's},
  descriptionplural={cumulant-generating functions},
  firstplural={\glsentrydescplural{cgf} (\glsentryplural{cgf})}
}
\newglossaryentry{pdf}
{
  name={PDF},
  description={probability density function},
  first={\glsentrydesc{pdf} (\glsentrytext{pdf})}
}
\newglossaryentry{gmi}
{
  name={GMI},
  description={generalized mutual information},
  first={\glsentrydesc{gmi} (\glsentrytext{gmi})}
}

\newglossaryentry{harqir}
{
  name={HARQ},
  description={Hybrid automatic repeat request},
  first={\glsentrydesc{harqir} (\glsentrytext{harqir})}
}

\newglossaryentry{harqcc}
{
  name={HARQ-CC},
  description={hybrid automatic repeat request with chase-combining},
  first={\glsentrydesc{harqcc} (\glsentrytext{harqcc})}
}

\newglossaryentry{fbl}
{
  name={FBL-NF},
  description={fixed blocklength and no feedback},
  first={\glsentrydesc{fbl} (\glsentrytext{fbl})}
}
\newglossaryentry{ssnf}
{
  name={SS-NF},
  description={single-shot no-feedback},
  first={\glsentrydesc{ssnf} (\glsentrytext{ssnf})}
}

\newglossaryentry{vlsf}
{
  name={VLSF},
  description={variable-length stop-feedback},
  first={\glsentrydesc{vlsf} (\glsentrytext{vlsf})}
}

\newglossaryentry{vlf}
{
  name={VLF},
  description={variable-length feedback},
  first={\glsentrydesc{vlf} (\glsentrytext{vlf})}
}

\newglossaryentry{cdf}
{
  name={CDF},
  description={cumulative distribution function},
  first={\glsentrydesc{cdf} (\glsentrytext{cdf})}
}

\newglossaryentry{ccdf}
{
  name={CCDF},
  description={complementary cumulative distribution function},
  first={\glsentrydesc{ccdf} (\glsentrytext{ccdf})}
}

\newglossaryentry{rce}
{
  name={RCEE},
  description={random-coding error-exponent},
  first={\glsentrydesc{rce} (\glsentrytext{rce})}
}

\newglossaryentry{grce}
{
  name={GRCEE},
  description={generalized random-coding error-exponent},
  first={\glsentrydesc{grce} (\glsentrytext{grce})}
}

\newglossaryentry{rcu}
{
  name={RCU},
  description={random-coding union},
  first={\glsentrydesc{rcu} (\glsentrytext{rcu})}
}

\newglossaryentry{qc}
{
  name={QC},
  description={quasi-cyclic},
  first={\glsentrydesc{qc} (\glsentrytext{qc})}
}

\newglossaryentry{wef}
{
  name={WEF},
  description={weight enumerator function},
  first={\glsentrydesc{wef} (\glsentrytext{wef})}
}

%% file: Config.tex
\newcommand{\TheThesisType}{licentiate of engineering} 
\newcommand{\TheAuthor}{Hao Guo}
\newcommand{\TheTitle}{Predictor Antenna Systems: Exploiting Channel State Information for Vehicle Communications} 
\newcommand{\TheSubtitle}{Abor am qui ipsa poristis ut faccaborum, temporio. Bo. Ic tor sit laboritius incto qui sit, utem eum sae none nis alitas maximperro quae}

\newcommand{\TheCoverPageSubtitle}{Abor am qui ipsa poristis ut faccaborum, temporio. Bo. Ic tor sit laboritius\\ incto qui sit, utem eum sae none nis alitas maximperro quae} 

\newcommand{\TheKeyword}{Beyond 5G, 6G, channel state information (CSI), hybrid automatic repeat request (HARQ), integrated access and backhaul (IAB), Marcum Q-function, mobility, mobile relay, outage probability, predictor antenna (PA), rate adaptation, relay, spatial correlation, temporal correlation, throughput, wireless backhaul.} 

\newcommand{\TheDepartment}{Electrical Engineering}
\newcommand{\TheYear}{2020}
\newcommand{\TheMonth}{May}
\newcommand{\TheDedication}{To my parents} 

\newcommand{\TheISSN}{1403-266X} 
\newcommand{\ThePrinter}{Chalmers Reproservice}

\newcommand{\CoverPageName}{\expandafter\MakeUppercase\expandafter{\TheAuthor}} 
\newcommand{\TitlePageName}{\TheAuthor} 
\newcommand{\TitlePageTitle}{\TheTitle} 
\newcommand{\ISSNPageName}{\TheAuthor} 
\newcommand{\ISSNPageTitle}{\TheTitle} 

\addpaper{Letter1}
 \addpaper{T1}
 \addpaper{Letter2}

\addnotincludedpaper%
    {A genetic algorithm-based beamforming approach for delay-constrained networks}%
    {\textbf{H. Guo}, B. Makki, and T. Svensson}%
    {in \textit{Proc. IEEE WiOpt}, Paris, France, May 2017, pp. 1-7}
    
\addnotincludedpaper%
    {A comparison of beam refinement algorithms for millimeter wave initial access}%
    {\textbf{H. Guo}, B. Makki, and T. Svensson}%
    {in \textit{Proc. IEEE PIMRCW}, Montreal, QC, Canada, Oct. 2017, pp. 1-7}

\addnotincludedpaper%
    {Genetic algorithm-based beam refinement for initial access in millimeter wave mobile networks}%
    {\textbf{H. Guo}, B. Makki, and T. Svensson}%
    {\textit{Wireless Communication and Mobile Computing}, Jun. 2018}

%% file: YourThesis/Abstract.tex
Vehicle communication is one of the most important use cases in the fifth generation of wireless networks (5G).  The growing demand for quality of service (QoS) characterized by performance metrics, such as spectrum efficiency, peak data rate,  and outage probability, is mainly limited by inaccurate prediction/estimation of channel state information (CSI) of the rapidly changing environment around moving vehicles. One way to increase the prediction horizon of CSI in order to improve the  QoS is  deploying predictor antennas (PAs).  A PA system consists of two sets of antennas typically mounted on the roof of a vehicle, where the PAs positioned at the front of the vehicle are used to predict the CSI observed by the receive antennas (RAs) that are aligned behind the PAs. In  realistic PA systems, however, the actual benefit is affected by a variety of factors, including spatial mismatch, antenna utilization, temporal correlation of scattering environment, and  CSI estimation error. This thesis investigates different resource allocation schemes for the PA systems under practical constraints, with main contributions summarized as follows.

First, in Paper A, we study the PA system in the presence of the so-called spatial mismatch problem, i.e., when the channel observed by the PA is not exactly the same as the one experienced by the RA. We derive closed-form expressions for the throughput-optimized rate adaptation, and evaluate the system performance in various temporally-correlated conditions for the scattering environment. Our results indicate that PA-assisted adaptive rate adaptation leads to a considerable performance improvement, compared to the cases with no rate adaptation. Then, to simplify e.g., various integral calculations as well as different operations such as parameter optimization, in Paper B, we  propose a semi-linear approximation of the Marcum Q-function, and apply the proposed approximation to the evaluation of the PA system. We also perform deep analysis of the effect of various parameters such as antenna separation as well as CSI estimation error. As we show, our proposed approximation scheme enables us to analyze PA systems with high accuracy. 

The second part of the thesis focuses on improving the spectral efficiency of the PA system by involving the PA into data transmission. In Paper C,  we analyze the outage-limited performance of PA systems using hybrid automatic repeat request (HARQ). With our proposed approach, the PA is used not only for improving the CSI in the retransmissions to the RA, but also for data transmission in the initial round.  As we show in the analytical and the simulation results, the combination of PA and HARQ protocols makes it possible to improve the  spectral efficiency  and adapt transmission parameters to mitigate the effect of spatial mismatch.

%% file: YourThesis/Acknowledgement.tex
\vspace{1cm}

As my PhD voyage is now half-way through, I would like to take the opportunity to recognize the people without whom this thesis would not have been possible.

First and foremost, I would like to express my deepest gratitude to Prof. Tommy Svensson for being my examiner and main supervisor and for giving me the opportunity to become a PhD student. Thank you for all the guidance, nice discussions and the constant support you have provided me for the last couple of years. This deepest gratitude also goes to my co-supervisor, Dr. Behrooz Makki, for over 1500 emails from you containing your fruitful comments and  detailed guidance, for the meetings you came all the way from Ericsson, for your valuable time even when you were super busy with the job hand-over and babies. I recommend everyone to collaborate with you and enjoy your kind, friendly, and pure personality.

Special thanks to Prof. Mohamed-Slim Alouini, for reviewing my draft and providing fruitful feedback and ideas. I would also like to thank Dr. Jingya Li for reading the rough draft of my paper, and all the help she has provided me outside of my research. I am also grateful to Prof. Mikael Sternad, Associate Prof. Carmen Botella, Prof. Xiaoming Chen, Dr. Fuxi Wen, and Dr. Nima Jamaly for all the nice discussions and collaborations we have had.

I would also like to thank the current and former members of the Communication Systems group. Many thanks to the head of our division, Prof. Erik Ström, and the head of our group, Prof. Fredrik Brännström, for ensuring a stimulating and joyful research atmosphere. Special thanks to Prof. Erik Agrell for our nice collaborations in the teaching work, and for everything you have shown me to be a kind and responsible teacher. Also, many thanks go to Professors Giuseppe Durisi, Henk Wymeersch, Alexandre Graell i Amat, Thomas Eriksson, and Jian Yang for the very interesting discussions in the courses. Gratitude also goes to Agneta,  Natasha,  Daniela,  Annie,  Yvonne, and  Madeleine for all your help. Thank you, Jinlin, Chao, Chenjie, Xinlin, Li, Yuxuan, Lei, Dapeng, Cristian, Johan, Keerthi, Anver, Shen, Roman, Chouaib, Mohammad, Björn, Sven, and Rahul for all the support and encouragement you have given to me.  I would also like to thank all my Chinese friends in Gothenburg for all the great moments we have experienced together. 

Special thanks to Yigeng, Xiao, Mengcheng, Pei, Hongxu, Qiang, and Kai for always being there for me. 

Finally, I would like to express my sincerest gratitude to mom and dad for your constant support, love, and encouragement over the years. I love you.

\vfill
\hfill   Hao Guo

\hfill  G\"{o}teborg, \TheMonth~\TheYear

\vfill

\clearpage

%% file: FinancialSupport.tex
This work was supported in part by VINNOVA (Swedish Government Agency for Innovation Systems) within the VINN Excellence Center ChaseOn, and in part by the EC within the H2020 project 5GCAR. The simulations were performed in part on resources provided by the Swedish National Infrastructure for Computing (SNIC) at C3SE.

%% file: Acronyms.tex
\begin{acronyms}
\item[2G/4G/5G/6G] Second/Fourth/Fifth/Sixth generation
\item[3GPP] 3rd generation partnership project
\item[ACK] Acknowledgment 
\item[ARQ/HARQ] Automatic repeat request/Hybrid automatic repeat request
\item[CDF] Cumulative distribution function
\item[CSI] Channel state information
\item[CSIT] Channel state information at the transmitter 
\item[DL] Downlink
\item[EMBB] Enhanced mobile broadband
\item[FDD] Frequency division duplex
\item[FSO] Free-space optical
\item[GPS] Global Positioning System
\item[IAB] Integrated access and backhaul
\item[i.i.d.] Identical and independently distributed
\item[INR]  Incremental redundancy
\item[LOS] Line-of-sight
\item[LTE] Long-Term Evolution
\item[MIMO] Multiple-input multiple-output
\item[MISO] Multiple-input single-output
\item[MRN] Moving relay node
\item[MTC] Machine-type communications
\item[NACK] Negative acknowledgment
\item[NLOS] Non-line-of-sight
\item[NMSE] Normalized mean squared error
\item[npcu] Nats-per-channel-use
\item[NR] New Radio
\item[OFDM] Orthogonal frequency-division multiplexing
\item[PA] Predictor antenna
\item[PDF] Probability density function
\item[QoS] Quality of service
\item[RA] Receive antenna
\item[RF] Radio-frequency
\item[RTD] Repetition time diversity
\item[SNR] Signal-to-noise ratio
\item[TDD] Time division duplex
\item[UL] Uplink
\item[URLLC] Ultra-reliable low-latency communications
\item[V2X] Vehicle-to-everything
\end{acronyms}

%% file: Content.tex

\input{YourThesis/chapters/One.tex}
\input{YourThesis/chapters/Two.tex}
\input{YourThesis/chapters/Three.tex}

\input{YourThesis/chapters/Conclusion.tex}


%% file: YourThesis/chapters/One.tex
\chapter{Introduction}
\label{ch:introduction}

\section{Background}
Nowadays, wireless communication and its related applications play  important roles in our life. Since the first mobile communication system employed in the early 1980s, new standards were established roughly every ten years, leading to the first commercial deployment of the \gls{5g} cellular networks in late 2019 \cite{ericsson2019,Dang2020what,andrews2014what}. From the  \gls{2g}, where the first digital communication system was deployed with text messages being available,  through the recent  \gls{4g} with \gls{3gpp} \gls{lte} being the dominant technology, to future \gls{5g} with \gls{nr} standardized by the \gls{3gpp} \cite{zaidi20185g}, one theme never changes: the growing demand for high-speed, ultra-reliable, low-latency and energy-efficient wireless communications with limited radio spectrum resource.

According to the Ericsson mobility report \cite{ericsson2019}, the total number of mobile subscriptions has exceeded 8.1 billion today,  with \gls{4g} being the major standard, and it is expected that this number will reach around 9 billion with over 20\% being supported by \gls{nr} by the end of 2024 \cite{ericsson2019}. Thanks to the higher bandwidth (usually larger than 1 GHz) at millimeter wave frequency spectrum, as well as the development of multi antenna techniques, new use cases in \gls{5g}, such as intelligent transport systems, autonomous vehicle control, virtual reality, factory automation, and providing coverage to high-mobility users, have been developed rapidly \cite{Simsek2016JSAC5g}. These use cases are usually categorized into three distinct classes by the standardization groups of \gls{5g} \cite{osseiran2014scenarios}:

\begin{itemize}
\item[i)]  \gls{emb} deals with large data packets and how to deliver them using high data rates \cite{dahlman20185g}. This can be seen as a natural extension of the current established \gls{lte} system that is designed for the similar use case. Typical  \gls{emb} applications involve high-definition video steaming, virtual reality, and online gaming.
\item[ii)] \gls{mtc} is a new application in \gls{5g}, which targets at providing wide coverage to a massive number of devices such as sensors who send sporadic updates to a \gls{bs} \cite{bockelmann2016massive}. Here, the key requirements are energy consumption, reliability, and scalability. High data rate and low latency, on the other hand, are of secondary importance. 
\item[iii)] \gls{urllc} concerns mission-critical applications with stringent requirements on reliability and latency \cite{bockelmann2016massive}. In this type of use case, the challenge is to design protocols which can transmit data with very low error probability and fulfill the latency constraint at the same time. Applications falling into this category include real-time control in smart factories, remote medical surgery, and \gls{v2x} communications which mainly focus on safety with high-mobility users. 
\end{itemize}

This thesis targets  both \gls{emb} and \gls{urllc}. More specifically, this work develops efficient (high data rate) and reliable (low error probability) \gls{v2x} schemes with latency requirement, using the \gls{pa} concept. A detailed review of the \gls{v2x} communications and  the \gls{pa} concept, as well as the associated research challenges, are presented in the following sub-sections.

\subsection{Vehicle Communications in 5G and Time/Space-Varying Channel}

Providing efficient, reliable broadband wireless communication links in high mobility use cases, such as high-speed railway systems, urban/highway vehicular communications, has been incorporated as an important part of the \gls{5g} developments \cite{imt2015}. According to \cite{samsung2015}, \gls{5g} systems are expected to support a large number of users traveling at speeds up to 500 km/h, at a data rate of 150 Mbps or higher. One interesting scenario in \gls{5g} vehicle communication is the \gls{mrn}, where a significant number of users could access cellular networks using moving relays, e.g., at public transportation such as buses, trams, and trains, via their smart phones or laptops \cite{yutao2013moving}. As one type of \gls{mrn}, one can consider the deployment of \gls{iab} nodes on top of the vehicles \cite{Teyeb2019VTCintegrated},  where part of the radio resources is used for wireless backhauling. In this way, moving \gls{iab} nodes can provide feasible solutions for such relay densification systems in \gls{5g}\footnote{It should be noted that mobile \gls{iab} is not supported in \gls{3gpp} Rel-16 and 17. However, it is expected to be discussed in the next releases.}.

Most current cellular systems can support users with low or moderate mobility, while high moving speed would  limit the coverage area and the data rate significantly. For example, \gls{4g} systems are aimed at supporting users perfectly at the speed of 0-15 km/h, serving with high performance from 15 km/h to 120 km/h, and providing functional services at 120-350 km/h \cite{3gpp2014requirements}.  On the other hand, field tests at different places \cite{Wu2016IEsurvey} have shown that current \gls{4g} systems can only provide 2-4 Mbps data rate in high-speed trains. To meet the requirement of high data rate at high moving speed in future mobility communication systems, new technologies that are able to cope with the challenges of mobility need to be developed.

With the setup of \gls{mrn} and other \gls{v2x} applications such as vehicle platooning \cite{guo2017wiopt,guo2017pimrc,guo2018hindawi} and remote driving \cite{Liang2017TVTvehiclar}, different technologies can be applied to improve the system performance at high speeds. For example, strategies in current standard aiming at improving the spectral efficiency include \gls{mimo}, \gls{csi}-based scheduling, and adaptive modulation and coding. Moreover, in the future standardization, techniques such as \gls{comp} \gls{jt} and massive \gls{mimo} will be also involved. All these techniques have one thing in common: they require accurate estimation of \gls{csit}  with acceptable cost. However, this is not an easy task. The main reason is that the channel in vehicle communication has certain features which makes it difficult to acquire \gls{csit} \cite{Wu2016IEsurvey}:
\begin{itemize}
\item[i)] \textit{Fast time-varying fading}: For high-speed vehicles, the channel has fast time-variation due to large Doppler spread. Let us consider a simple example. Assume a vehicle operating at a speed of 200 km/h and a frequency of 6 GHz. Then, the maximum Doppler frequency is obtained by $f_\text{D} = v/\lambda = $1111 Hz, which corresponds to a channel coherence time of around 900 $\mu$s. However, in \gls{lte} the control loop time with both \gls{ul} and \gls{dl} is around 2 ms, which makes \gls{csit} outdated if we consider the \gls{tdd} system with channel reciprocity. Moreover, the speeds of moving terminals are usually time-varying, making the channel even more dynamic.
\item[ii)] \textit{Channel estimation errors}: Due to the time-varying channel, it is not practical to assume perfect \gls{csit}, as we do for low mobility systems. In fact, mobility causes difficulties not only on accurately estimating the channel, but also on tracking, updating and predicting the fading parameters. Also, the estimation error may have remarkable effects on system performance, which makes this aspect very important  in the system design.
\item[iii)] \textit{Doppler diversity}: Doppler diversity has been developed for systems with perfect \gls{csit}, in which it provides diversity gain to improve system performance. On the other hand, Doppler diversity may cause high channel estimation error, which makes it important to study the trade-off between Doppler diversity and estimation errors.
\end{itemize} 
Besides these three aspects, there are also some issues for the channel with mobility, e.g., carrier frequency offset, inter-carrier interference, high penetration loss, and frequent handover. To conclude, with the existing methods and depending on the vehicle speed, channel coefficients may be outdated at the time of transmission, due to various delays in the control loop and the mobility of the vehicles.

The use of channel predictions can alleviate this problem. By using the statistics over time and frequency, combining with linear predictors such as Kalman predictor, the channel coefficients can be predicted for around 0.1-0.3 carrier wavelengths in space \cite{Sternad2012WCNCWusing}. This prediction horizon is enough for \gls{4g} systems with short control loops (1-2 ms) or for  users with pedestrian velocities. However, it is inadequate for vehicular velocities at high frequencies.

\subsection{Predictor Antenna and Related Work}
To overcome the issue of limited prediction horizon in the rapidly changed channel with mobility, and to support use cases such as \gls{mrn}, \cite{Sternad2012WCNCWusing} proposed the concept of \gls{pa}. Here, the \gls{pa} system refers to  a setup with two sets of antennas on the roof of a vehicle, where the \glspl{pa} positioned in the front of the vehicle are used to predict the channel state observed by one \gls{ra} or a set of \glspl{ra} that are aligned behind the \glspl{pa}, and send the \gls{csi} back to the \gls{bs}. Then, if the \gls{ra} reaches the same point as the \gls{pa}, the \gls{bs} can use the  \gls{csi} obtained from the \glspl{pa} to improve the transmission to the \glspl{ra} using, for example, power/rate adaptations and beamforming. The results in \cite{Sternad2012WCNCWusing} indicate that the \gls{pa} system can provide sufficiently accurate  channel estimation for at least one wavelength in the \gls{los} case, and \cite{Jamaly2014EuCAPanalysis} shows that with a smoothed roof of the vehicle to avoid refraction, abnormal reflection and scattering, and  with antenna coupling compensation at least 3 wavelengths can be predicted in both \gls{los} and \gls{nlos} conditions.

Following \cite{Sternad2012WCNCWusing}, \cite{BJ2017ICCWusing,phan2018WSAadaptive,Apelfrojd2018PIMRCkalman} provide experimental validation to prove the feasibility of the \gls{pa} concept. Specifically, \cite{BJ2017ICCWusing} presents an order of magnitude increase of prediction horizons compared to time-series-based prediction. Moreover, \cite{phan2018WSAadaptive} shows that the  \gls{pa} concept works for massive \gls{mimo} \glspl{dl} where the \gls{pa} can improve the signal-to-interference ratio in setups  with \gls{nlos} channels. Also, \cite{Apelfrojd2018PIMRCkalman} demonstrates that the Kalman smoothing-based \gls{pa} system enables up to 0.75 carrier wavelengths prediction at vehicle speeds for Rayleigh-like \gls{nlos} fading channels. The review of \cite{Sternad2012WCNCWusing,BJ2017ICCWusing,phan2018WSAadaptive,Apelfrojd2018PIMRCkalman} reveals the following research problems:
\begin{itemize}
\item[i)] \textit{Speed sensitivity}: From the results in \cite{BJ2017ICCWusing,phan2018WSAadaptive,Apelfrojd2018PIMRCkalman}, we can observe that, for given control loop time, if the speed is too low or too high which leads to large distances, i.e., spatial mismatch,  between the spot where the \gls{pa} estimates the channel and the spot where the \gls{ra} reaches at the second time slot, the accuracy of prediction decreases drastically. We cannot make sure that the speed of the vehicle remains the same all the time, which may lead to performance loss. Indeed, \cite{DT2015ITSMmaking} and \cite{Jamaly2019IETeffects} have addressed this kind of spatial mismatch problem in the \gls{pa} system. In \cite{DT2015ITSMmaking}, an interpolation-based beamforming scheme is proposed for \gls{dl} \gls{miso} systems to solve the mis-pointing problem. From another perspective, \cite{Jamaly2019IETeffects} studies the effect of  velocity variation on  prediction performance. However, how to analytically study  speed sensitivity of  the \gls{pa} system remains unclear.
\item[ii)] \textit{Lack of analytical model}: As we can see, \cite{Sternad2012WCNCWusing,BJ2017ICCWusing,phan2018WSAadaptive,Apelfrojd2018PIMRCkalman} are based on real-world testing data which validates the concept, while \cite{DT2015ITSMmaking} is based on simulated channel, and \cite{Jamaly2019IETeffects}  focuses more on the antenna pattern. No analytical model of the \gls{pa} system has been proposed in \cite{Sternad2012WCNCWusing,BJ2017ICCWusing,phan2018WSAadaptive,Apelfrojd2018PIMRCkalman,DT2015ITSMmaking,Jamaly2019IETeffects}. Moreover, as mentioned in the previous item, we need an analytical tool to study the sensitivity of the system performance to speed variation.
\item[iii)] \textit{What else can we do with the \gls{pa} system}: As we can see from the results in \cite{Sternad2012WCNCWusing,BJ2017ICCWusing,phan2018WSAadaptive,Apelfrojd2018PIMRCkalman,DT2015ITSMmaking,Jamaly2019IETeffects}, although the \gls{pa} system can provide larger prediction horizons for up to three wavelengths, there is still a limit on the region, and the system is quite sensitive to vehicle speed. Hence, additional structure/schemes could potentially be built on top of the \gls{pa} system to achieve better performance.
\item[iv)]  \textit{When to use the \gls{pa} system}: The key point of the \gls{pa} concept is to use an additional antenna to acquire better quality of \gls{csit}. In this way, the time-frequency resources of the \gls{pa} are used for channel prediction instead of data transmission. Intuitively, there should exist a condition under which the \gls{pa} concept could be helpful, compared to the case with simply using the \gls{pa} as one of the \glspl{ra}. Here, theoretical models may help us  make such decisions.
\end{itemize}



\section{Scope of the Thesis}
The aim of this thesis is to present analytical evaluation of the \gls{pa} system and, at the same time, to apply some key-enablers of \gls{urllc}, such as rate adaptation, \gls{harq}, and power allocation, considering imperfect \gls{csi} estimation.
The channel considered in this thesis is the non-central Chi-square distributed fading channel, which we model as the combination of the known part of the channel from the \gls{pa}, and the uncertainty part from the spatial mismatch. Firstly, in Paper A,  we present our proposed analytical model for evaluating the sensitivity of the \gls{pa} system with spatial mismatch. Some preliminary work on how to use rate adaptation based on  imperfect \gls{csi} is also presented. 

In Paper B, we first develop a mathematical tool that can be used to remarkably simplify the analysis of our proposed channel, some integral calculations, as well as optimization problems that contain the first-order Marcum $Q$-function. Then, we extend the work in Paper A and perform deep analysis of the effect of various parameters, such as processing delay of the \gls{bs} and imperfect feedback schemes. 

Besides the results in Paper A and B, we are also interested in how to further exploit the \gls{pa} system by, e.g., involving the \gls{pa} partly into the transmission process.
In Paper C, we propose an \gls{harq}-based \gls{pa} system which uses the \gls{bs}-\gls{pa} link for the initial transmission, and the feedback bit on the decoding results combined with the \gls{csi} estimation  for adapting the transmission parameters during  the   \gls{bs}-\gls{ra} transmission. Moreover, we develop power allocation schemes based on the \gls{harq}-\gls{pa} structure and study the outage-constrained average power of the system.

The specific objectives of this thesis can be summarized as follows.
\begin{itemize}
\item[i)] To characterize the speed sensitivity of the \gls{pa} system by analytically modeling the channels in  \gls{pa} systems.
\item[ii)] To develop a mathematical tool in order to simplify the performance evaluation of the \gls{pa} setup which involves the Marcum $Q$-function. 
\item[iii)] To design efficient and reliable transmission schemes which are able to improve the performance of existing \gls{pa} systems.
\end{itemize}

\section{Organization of the Thesis}
In Chapter 2, we introduce the \gls{pa} setups that are considered in the thesis. Specifically, we model the spatial mismatch in the \gls{pa} system and define the data transmission  model. The details of the channel model which involve the Marcum Q-function are also presented. To help the analytical evaluations, we provide a review of the use cases of the Marcum Q-function in a broad range of research areas, and present our proposed semi-linear approximation of the Marcum Q-function, with its applications on integral calculations and optimizations. In Chapter 3,  we present different resource allocation schemes, namely, rate adaptation and \gls{harq}-based power allocation, to improve the performance of the \gls{pa} system under the mismatch problem. For each scheme, we show the problem formulation, the data transmission model as well as the details of the proposed method. Finally, in Chapter 4, we provide a brief overview of our contributions in the attached papers, and discuss possible future research directions.


%% file: YourThesis/chapters/Two.tex
\chapter{PA Systems and Analytical Channel Model}\label{chapter:two}
This chapter first introduces the PA concept in a \gls{tdd} \footnote{The \gls{pa} concept can be applied in both \gls{tdd} and \gls{fdd} systems. In \gls{fdd} the \gls{pa} estimates the \gls{dl} channel based on \gls{dl} pilots from the \gls{bs}, and reports back using an \gls{ul} feedback channel. The \gls{bs} uses this information (as input) to obtain the \gls{dl} channel estimate to be used for the \gls{dl} towards the \gls{ra} when the \gls{ra} reaches the same spatial point as the \gls{pa} at the time of \gls{dl} estimation. On the other hand, in \gls{tdd} the \gls{pa} instead sends the pilots and the \gls{bs} estimates the \gls{ul} channel, and uses that in combination with channel reciprocity information (as input) to obtain the \gls{dl} channel estimate to be used for the \gls{dl} towards the \gls{ra}.} \gls{dl} \footnote{This thesis mainly focus on the \gls{dl}, but the \gls{pa} concept can be adapted to be used also in the \gls{ul} case.} system, where one \gls{pa} and one \gls{ra} are deployed at the receiver side. Also, the associated challenges and difficulties posed by practical constraints are discussed. Then, the proposed analytical channel model based on Jake's assumption is presented, where the \gls{cdf} of the channel gain is described by the first-order Marcum $Q$-function. Finally, to simplify the analytical derivations, we develop a semi-linear approximation of the first-order Marcum $Q$-function which can simplify, e.g., integral calculations as well as optimization problems.

\section{The PA Concept}
In \gls{5g}, a significant number of users would access wireless networks in vehicles, e.g., in public transportation like trams and trains or private cars, via their smart phones and laptops \cite{shim2016traffic,lannoo2007radio,wang2012distributed,Dat2015,Laiyemo2017,yutao2013moving,Marzuki2017,Andreev2019,Haider2016,Patra2017}. In \cite{shim2016traffic}, the emergence of vehicular heavy user traffic is observed by field experiments conducted in 2012 and 2015 in Seoul, and the experimental results reveal that such traffic is becoming dominant, as shown by the 8.62 times increase from 2012 to 2015 in vehicular heavy user traffic, while total traffic increased only by 3.04 times. Also, \cite{lannoo2007radio,wang2012distributed,Dat2015,Laiyemo2017} develop traffic schemes and networks for users in high-speed trains. Setting an \gls{mrn} in vehicles can be one promising solution to provide a high-rate reliable connection between a \gls{bs} and the users inside the vehicle \cite{yutao2013moving,Marzuki2017,Andreev2019}. From another perspective, \cite{Haider2016} and \cite{Patra2017} adopt femtocell technology inside a vehicle to provide better spectral and energy efficiency compared to the direct transmission scheme.

In such a so-called hot spot scenario, we often deploy \gls{tdd} systems with channel reciprocity. It is  intuitively because we have more data in \gls{dl} than in \gls{ul}. Here, we estimate the \gls{dl} channel based on the \gls{ul} pilots. Then, the problem may occur because of the movement, and the channel in the \gls{dl} would not be the same as the one in the \gls{ul}. This could be compensated for by extrapolating the \gls{csi} from the \gls{ul}, for example by using Kalman predictions \cite{Aronsson2011}. However, it is difficult to predict small-scale fading by extrapolating the past estimates and the prediction horizon is limited to 0.1-0.3$\lambda$ with $\lambda$ being the carrier wavelength \cite{Ekman2002}. Such a horizon is satisfactory for pedestrian users, while for high mobility users, such as vehicles, a prediction horizon beyond 0.3$\lambda$ is usually required \cite{Apelfrojd2018PIMRCkalman}. One possible way to increase the prediction horizon is to have a database of pre-recorded coordinate-specific \gls{csi}  at the \glspl{bs} \cite{zirwas2013channel}. Here, the basic idea is that the users provide the \glspl{bs} with their location information by, e.g., \gls{gps}, and the \gls{bs} could use the pre-recorded information to predict the channel environment. However, such a method requires large amount of data which may need to be updated frequently, and \gls{gps} position data would also not be accurate for small scale fading prediction, since the accuracy is much worse than a wavelength for typical mobile communications systems.


\begin{figure}
    \centering
    \includegraphics[width=0.7\columnwidth]{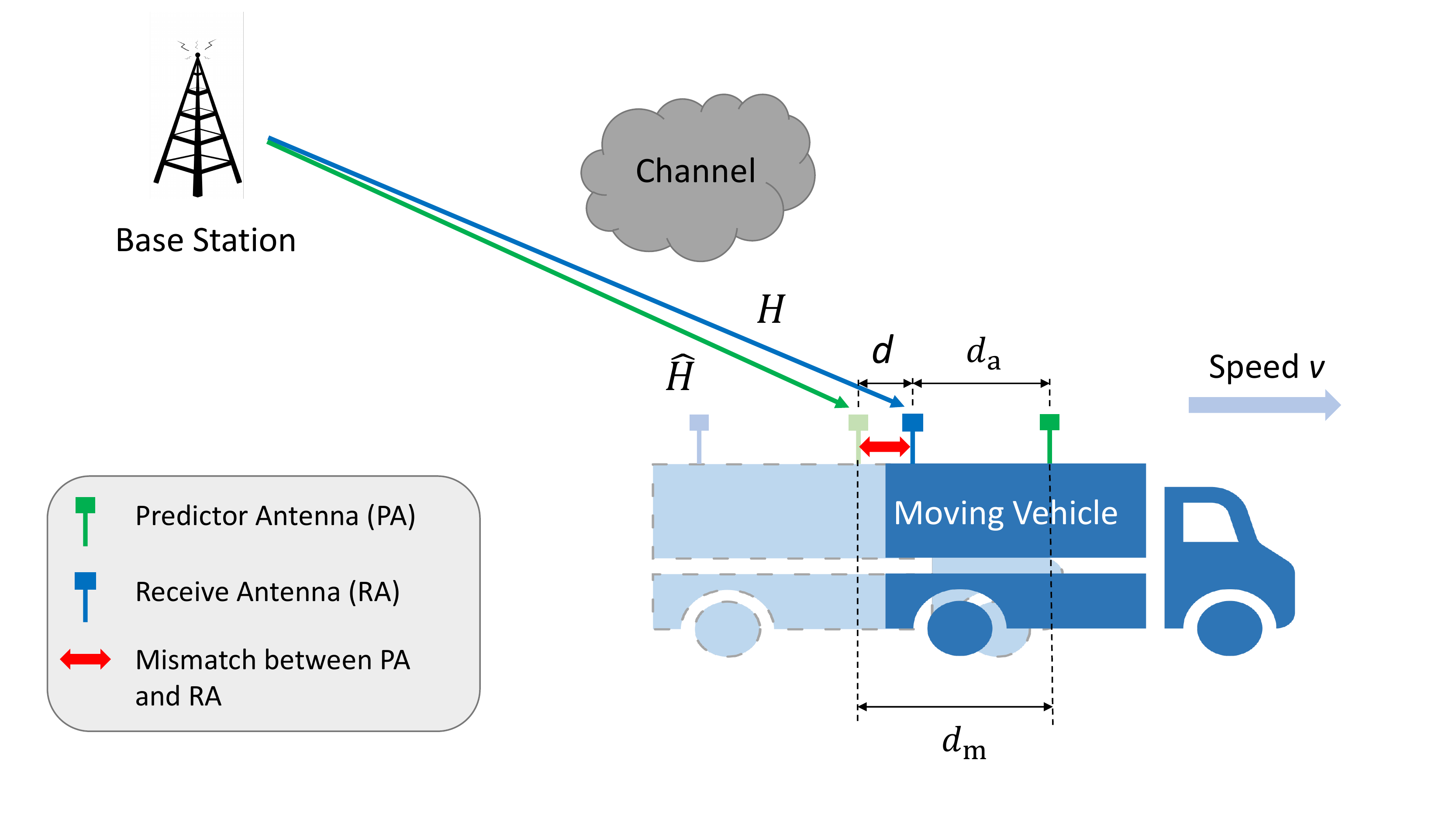}
    \caption{The PA concept with spatial mismatch problem.}
    \label{fig:PA_mismatch}
\end{figure}

To overcome this issue,  \cite{Sternad2012WCNCWusing} proposes the concept of \gls{pa} wherein at least two antennas are deployed on top of the vehicle.  As can be seen from Fig. \ref{fig:PA_mismatch}, the first antenna, which is the \gls{pa}, estimates the channel $\hat{H}$ in the \gls{ul} to the \gls{bs}. Then, the \gls{bs} uses the information received about $\hat H$ to estimate the channel $H$, and communicate with a second antenna, which we refer to as the \gls{ra}, when it arrives to the same position as the \gls{pa}. Then, a problem appears: how should we model such a channel? The intuitive idea is that the correlation between $H$ and $\hat{H}$ should be affected by the moving speed $v$, the time for \gls{ul} and \gls{dl}, as well as the antenna separation $d_\text{a}$ between the \gls{pa} and the \gls{ra}.

One way to evaluate such a model is to measure $H$ and $\hat{H}$ under different system configurations, and calculate the \gls{nmse} of $H$ and $\hat{H}$. Followed by \cite{Sternad2012WCNCWusing}, experimental results in \cite{BJ2017ICCWusing} and \cite{BJ2017PIMRCpredictor}  show that an \gls{nmse} of around -10 dB can be obtained for speeds up to 50 km/h, with all measured predictions horizons up to 3$\lambda$, which is ten times longer than the limit for Kalman filter-based channel extrapolation.  In \cite{BJ2017ICCWusing,BJ2017PIMRCpredictor,phan2018WSAadaptive} \gls{fdd} systems are considered, where dense enough \gls{dl} channel estimation pilots with \gls{ofdm} are used. On the other hand, for \gls{tdd} systems, the \gls{ul} and \gls{dl} frames need to be adjusted so that the estimation of $H$ can be as close as possible to $\hat{H}$, as proposed and evaluated in \cite{DT2015ITSMmaking}. However, such a method would need to adapt \gls{ul} and \gls{dl} ratios for each user, which is complicated from system design point of view.  To mitigate this issue, \cite{DT2015ITSMmaking} also proposes an interpolation scheme at the \gls{bs} which is suitable for different \gls{ul} and \gls{dl} ratios. Also, a Kalman smoother for the interpolation of the \gls{pa} for the \gls{tdd} case with a two-filter approach is proposed in \cite{Apelfrojd2018PIMRCkalman}, where the \gls{csi} quality of the \gls{dl} can be improved such that the duration of the \gls{dl} can be extended remarkably. Moreover, it is shown that the correlation between  $H$ and $\hat{H}$ would be reduced, if the \gls{pa} and the \gls{ra} are too close to each other, e.g., 0.2-0.4$\lambda$. Different ways to compensate for such a coupling effect, such as open-circuit decoupling,  are proposed in \cite{Jamaly2014EuCAPanalysis, Jamaly2019IETeffects}.

\section{Challenges and Difficulties}
Previous studies have shown that deploying the \gls{pa} system can provide significant performance gains in terms of, e.g., \gls{nmse} \cite{BJ2017ICCWusing,BJ2017PIMRCpredictor,phan2018WSAadaptive}.  However,  realistic gains can be limited by many practical constraints. In this section, we discuss a number of such challenges that have been partly addressed in this work.

\subsection*{Lack of Analytical Model}

In the literature, most \gls{pa} work rely on experimental measurements and simulations. This is sufficient for the validation purpose. However, to have a deeper understanding of  the \gls{pa} system, it is useful to develop analytical models. There are different statistical wireless channel models such as Rayleigh, Rice, Nakagami, and log-normal fading, as well as their combinations on multi-path and shadow fading components \cite{vatalaro1995generalized,tjhung1999fade}. Here, obtaining an exact analytical model for the \gls{pa} system may be difficult, but understanding the correlation between $H$ and $\hat{H}$ would be a good starting point.

\subsection*{Spatial Mismatch}
As addressed in, e.g., \cite{DT2015ITSMmaking,Jamaly2019IETeffects}, even assuming that the channel does not change over time, if the \gls{ra} does not arrive at the same point as the \gls{pa}, the actual channel for the \gls{ra} would not be identical to the one experienced by the PA before. As can be seen in Fig. \ref{fig:PA_mismatch} with \gls{tdd} setup, considering one vehicle deploying two antennas on the roof with one \gls{pa} positioned in the front of the moving direction and an \gls{ra} aligned behind the \gls{pa}. The idea of the data transmission model with \gls{tdd} is that the \gls{pa} first sends pilots at time $t$, then the \gls{bs} estimates the channel and sends the data at time $t+\delta$ to the \gls{ra}. Here, $\delta$ depends on the processing time at the \gls{bs}. Then, we define $d$ as the effective distance between the position of the PA at time $t$ and the position of the RA at time $t+\delta$, as can be seen in Fig. \ref{fig:PA_mismatch}. That is,  $d$ is given  by
\begin{align}\label{eq_d}
    d = |d_\text{a} - d_\text{m} | = |d_\text{a} - v\delta|,
\end{align}
where $d_\text{m}$ is the moving distance between  $t$ and $t+\delta $ while $v$ is the velocity of the vehicle. To conclude, different values of $v$, $ \delta$, $f_\text{c}$ and $d_\text{a}$ in (\ref{eq_d}) correspond to different values of $d$. We would like to find out how to connect $H$ and $\hat{H}$ as a function of $d$, and how different values of $d$ would affect the system performance.

\subsection*{Spectral Efficiency Improvement}
In a typical \gls{pa} setup, the spectrum is underutilized, and the spectral efficiency could be further improved in case the \gls{pa} could be used not only for channel prediction  but also for data transmission. However, proper data transmission schemes need to be designed to make the best use of the \gls{pa}.

\subsection*{Temporal Correlation}
The overhead from the \gls{ul}-\gls{dl} structure of the \gls{pa} system would affect the accuracy of the \gls{csi} acquisition, i.e., $\hat{H}$ obtained from \gls{pa} would change over time. Basically, the slowly-fading channel is not always a realistic model for fast-moving users, since the channel may change according to the environmental effects during a transmission block \cite{tsao2007prediction,Makki2013TCfeedback,Makki2014ISWCSreinforcement}. There are different ways to model the temporally-correlated channel, such as using the first-order Gauss-Markov process \cite{Makki2013TCfeedback,Makki2014ISWCSreinforcement}.

\subsection*{Estimation Error}
There could  be channel estimation errors from the \gls{ul} \cite{jose2011pilot}, which would degrade the system performance. The assumption of perfect channel reciprocity in \gls{tdd} ignores two important facts \cite{Mi2017massive}: 1) the \gls{rf} chains of the \gls{ul} and the \gls{dl} are separate circuits with random impacts on the transmitted and received signals \cite{Mi2017massive,lu2014an}, which is the so-called \gls{rf} mismatch; 2) the interference profile at the transmitter and receiver sides are different \cite{tolli2017compensation}. These deviations are defined as reciprocity errors that invalidate the assumption of perfect reciprocity, and should be considered in the system design.

\subsection*{Effects of Other Parameters}
As mentioned in Fig. \ref{fig:PA_mismatch}, different system parameters such as the speed $v$, the antenna separation $d_\text{a}$ and the control loop time $\delta$ would affect the system behaviour by, e.g., spatial mismatch or antenna coupling. Our goal is to study the effect of these parameters and develop robust schemes which perform well for a broad range of their values.

\section{Analytical Channel Model}
Considering \gls{dl} transmission in the \gls{bs}-\gls{ra} link, which is our main interest, the received signal is given by \footnote{In this work, we mainly focus on the cases with single \gls{pa} and \gls{ra} antennas. The future works will address the problem in the cases with array antennas.}
\begin{align}\label{eq_Y}
{{Y}} = \sqrt{P}HX + Z.
\end{align}
Here, $P$ represents the average received power at the \gls{ra}, while $X$ is the input message with unit variance, and $H$ is the fading coefficient between the \gls{bs} and the \gls{ra}. Also, $Z \sim \mathcal{CN}(0,1)$ denotes the \gls{iid} complex Gaussian noise added at the receiver.

We denote the channel coefficient of the \gls{pa}-\gls{bs} \gls{ul} by $\hat{H}$ and we assume that $\hat{H}$ is perfectly known by the \gls{bs}. The result can be extended to the cases with imperfect \gls{csi} at the BS (see our work \cite{guo2020semilinear}). In this way,  we use the spatial correlation model \cite[p. 2642]{Shin2003TITcapacity}
\begin{align}\label{eq_tildeH}
    \Tilde{\bm{H}} = \bm{\Phi}^{1/2} \bm{H}_{\varepsilon},
\end{align}
where $\Tilde{\bm{H}}$ = $\bigl[ \begin{smallmatrix}
  \hat{H}\\H
\end{smallmatrix} \bigr]$ is the channel matrix including both \gls{bs}-\gls{pa} channel $\hat{H}$ and \gls{bs}-\gls{ra} channel $H$ links. $\bm{H}_{\varepsilon}$ has independent circularly-symmetric zero-mean complex Gaussian entries with unit variance, and $\bm{\Phi}$ is the channel correlation matrix.

In general, the spatial correlation of the fading channel depends on the distance between the \gls{ra} and the \gls{pa}, which we denote by $d_\text{a}$, as well as the angular spectrum of the radio wave pattern. If we use the classical Jakes' correlation model by assuming uniform angular spectrum, the $(i,j)$-th entry of $\bm{\Phi}$ is given by \cite[Eq. 1]{Chizhik2000CLeffect}
\begin{align}\label{eq_phi}
    \Phi_{i,j} = J_0\left((i-j)\cdot2\pi d/ \lambda\right).
\end{align}
Here, $J_0(\cdot)$ is the zeroth-order Bessel function of the first kind. Also, $\lambda = c/f_\text{c}$ represents the wavelength  where $c$ is the speed of light and $f_\text{c}$ is the carrier frequency. 

As discussed before, different values of $v$, $ \delta$, $f_\text{c}$ and $d_\text{a}$ in (\ref{eq_d}) correspond to different values of $d$, which leads to different levels of channel spatial correlation~(\ref{eq_tildeH})-(\ref{eq_phi}).

Combining (\ref{eq_tildeH}) and (\ref{eq_phi}) with normalization,  we have
\begin{align}\label{eq_H}
    H = \sqrt{1-\sigma^2} \hat{H} + \sigma q,
\end{align}
where $q \sim \mathcal{CN}(0,1)$ which is independent of the known channel value $\hat{H}\sim \mathcal{CN}(0,1)$, and $\sigma$ is a function of the mismatch distance $d$.

From (\ref{eq_H}), for a given $\hat{H}$ and $\sigma \neq 0$, $|H|$ follows a Rician distribution, i.e., the \gls{pdf} of $|H|$ is given by 
\begin{align}
  f_{|H|\big|\hat{H}}(x) = \frac{2x}{\sigma^2}e^{-\frac{x^2+(1-\sigma^2)\hat{g}}{\sigma^2}}I_0\left(\frac{2x\sqrt{(1-\sigma^2)\hat{g}}}{\sigma^2}\right),   
\end{align}
with $\hat{g} = |{\hat{H}}|^2$, and $I_n(x) = (\frac{x}{2})^n \sum_{i=0}^{\infty}\frac{(\frac{x}{2})^{2i} }{i!\Gamma(n+i+1)}$ being the $n$-th order modified Bessel function of the first kind where $\Gamma(z) = \int_0^{\infty} x^{z-1}e^{-x} \mathrm{d}x$ denotes the Gamma function. Then, we define the channel gain between BS-RA as $ g = |{H}|^2$. By changing variables from $H$ to $g$, the \gls{pdf} of $f_{g|\hat{H}}$ is given by

\begin{align}\label{eq_pdf}
    f_{g|\hat{H}}(x) = \frac{1}{\sigma^2}e^{-\frac{x+(1-\sigma^2)\hat{g}}{\sigma^2}}I_0\left(\frac{2\sqrt{x(1-\sigma^2)\hat{g}}}{\sigma^2}\right),
\end{align}
which is non-central Chi-squared distributed, and the \gls{cdf} is
\begin{align}\label{eq_cdf}
    F_{g|\hat{H}}(x) = 1 - Q_1\left( \sqrt{\frac{2(1-\sigma^2)\hat{g}}{\sigma^2}}, \sqrt{\frac{2x}{\sigma^2}}  \right).
\end{align}
Here,  $Q_1(\alpha,\beta)$ is Marcum $Q$-function and it is defined as  \cite[Eq. 1]{Bocus2013CLapproximation}
\begin{align}\label{eq_q1}
    Q_1(\alpha,\beta) = \int_{\beta}^{\infty} xe^{-\frac{x^2+\alpha^2}{2}}I_0(x\alpha)\text{d}x,
\end{align}
where $\alpha, \beta \geq 0$.

We study the system performance in various temporally-correlated conditions, i.e., when $H$ is not the same as $\hat{H}$ even at the same position. Particularly, using the same model as in  \cite[Eq. 2]{Makki2013TCfeedback}, we further develop our channel model  (\ref{eq_H}) as
\begin{align}\label{eq_Htp}
    H_{k+1} = \beta H_{k} + \sqrt{1-\beta^2} z, 
\end{align}
for each time slot $k$, where $z \sim \mathcal{CN}(0,1)$ is a Gaussian noise which is uncorrelated with $H_{k}$. Also, $\beta$ is a known correlation factor which represents two successive channel realizations dependencies by $\beta = \frac{\mathbb{E}\{H_{k+1}H_{k}^*\}}{\mathbb{E}\{|H_k|^2\}}$. Substituting (\ref{eq_Htp}) into (\ref{eq_H}), we have
\begin{align}\label{eq_Ht}
    H_{k+1} = \beta\sqrt{1-\sigma^2}\hat{H}_{k}+\beta\sigma q+\sqrt{1-\beta^2}z = \beta\sqrt{1-\sigma^2}\hat{H}_{k} + w.
\end{align}
Here, to simplify the calculation,  $\beta\sigma q + \sqrt{1-\beta^2}z$ is equivalent to a new Gaussian variable $w \sim\mathcal{CN}\left(0,(\beta\sigma)^2+1-\beta^2\right)$.  Moreover, we can follow the same approach as in \cite{Wang2007TWCperformance} to add the effect of estimation errors of $\hat{H}$ as an independent additive Gaussian variable whose variance is given by the accuracy of \gls{csi} estimation.

\section{The First-Order Marcum Q-Function and  Semi-Linear Approximation}
The first-order \footnote{To simplify the analysis, our work concentrates on the approximation of the first-order Marcum-$Q$ function. However, our approximation technique can be easily extended to the cases with different orders of the Marcum $Q$-function.} Marcum $Q$-function (\ref{eq_q1}) is observed in various problem formulations. However, it is not an easy-to-handle function with modified Bessel function, double parameters ($\alpha$ and $\beta$), and the integral shape.

In the literature,  the Marcum $Q$-function has appeared in many areas, such as statistics/signal detection \cite{helstrom1994elements}, and in performance analysis of different setups, such as temporally correlated channels \cite{Makki2013TCfeedback}, spatial correlated channels \cite{Makki2011Eurasipcapacity},  free-space optical (FSO) links \cite{Makki2018WCLwireless}, relay networks \cite{Makki2016TVTperformance}, as well as  cognitive radio and radar systems \cite{Simon2003TWCsome,Suraweera2010TVTcapacity,Kang2003JSAClargest,Chen2004TCdistribution,Ma2000JSACunified,Zhang2002TCgeneral,Ghasemi2008ICMspectrum,Digham2007TCenergy, simon2002bookdigital,Cao2016CLsolutions,sofotasios2015solutions,Cui2012ELtwo,Azari2018TCultra,Alam2014INFOCOMWrobust,Gao2018IAadmm,Shen2018TVToutage,Song2017JLTimpact,Tang2019IAan,ermolova2014laplace,peppas2013performance,jimenez2014connection}. However, in these applications, the presence of  the Marcum $Q$-function makes the mathematical analysis challenging, because it is  difficult to manipulate  with no closed-form  expressions especially when it appears in parameter optimizations and integral calculations. For this reason, several methods have been developed in \cite{Bocus2013CLapproximation,Fu2011GLOBECOMexponential,zhao2008ELtight,Simon2000TCexponential,annamalai2001WCMCcauchy,Sofotasios2010ISWCSnovel,Li2010TCnew,andras2011Mathematicageneralized,Gaur2003TVTsome,Kam2008TCcomputing,Corazza2002TITnew,Baricz2009TITnew,chiani1999ELintegral}  to bound/approximate the Marcum $Q$-function. For example, \cite{Fu2011GLOBECOMexponential,zhao2008ELtight} have proposed  modified forms of the function, while \cite{Simon2000TCexponential,annamalai2001WCMCcauchy} have derived exponential-type bounds which are good for the bit error rate analysis at high signal-to-noise ratios (SNRs). Other types of bounds are expressed by, e.g., error function \cite{Kam2008TCcomputing} and Bessel functions \cite{Corazza2002TITnew,Baricz2009TITnew,chiani1999ELintegral}. Some alternative methods have been also proposed in \cite{Sofotasios2010ISWCSnovel,Li2010TCnew,andras2011Mathematicageneralized,Bocus2013CLapproximation,Gaur2003TVTsome}. Although each of these approximation/bounding techniques are fairly tight for their considered problem formulation, they are still based on hard-to-deal functions, or have complicated summation/integration structures, which may be not easy to deal with in e.g., integral calculations and parameter optimizations. 

We present our semi-linear approximation of the \gls{cdf} in the form of $y(\alpha, \beta ) = 1-Q_1(\alpha,\beta)$. The idea of this proposed approximation is to use one point and its corresponding slope at that point to create a line  approximating the CDF. The approximation method is summarized in Lemma \ref{Lemma1} as follows.
\begin{lem}\label{Lemma1}
 The CDF of the form $y(\alpha, \beta ) = 1-Q_1(\alpha,\beta)$ can be semi-linearly approximated by $Y(\alpha,\beta)\simeq \mathcal{Z}(\alpha, \beta)$ where
\begin{align}\label{eq_lema1}
\mathcal{Z}(\alpha, \beta)=
\begin{cases}
0,  ~~~~~~~~~~~~~~~~~~~~~~~~~~~~~~~~~~~\mathrm{if}~ \beta < c_1  \\ 
 \frac{\alpha+\sqrt{\alpha^2+2}}{2} e^{-\frac{1}{2}\left(\alpha^2+\left(\frac{\alpha+\sqrt{\alpha^2+2}}{2}\right)^2\right)}\times\\
 ~~~I_0\left(\alpha\frac{\alpha+\sqrt{\alpha^2+2}}{2}\right)\times\left(\beta-\frac{\alpha+\sqrt{\alpha^2+2}}{2}\right)+\\
 ~~~1-Q_1\left(\alpha,\frac{\alpha+\sqrt{\alpha^2+2}}{2}\right),  ~~~~~~\mathrm{if}~ c_1 \leq\beta\leq c_2 \\
1,  ~~~~~~~~~~~~~~~~~~~~~~~~~~~~~~~~~~~\mathrm{if}~ \beta> c_2,
\end{cases}
\end{align}
with
\begin{align}\label{eq_c1}
  c_1(\alpha) = \max\Bigg(0,\frac{\alpha+\sqrt{\alpha^2+2}}{2}+
    \frac{Q_1\left(\alpha,\frac{\alpha+\sqrt{\alpha^2+2}}{2}\right)-1}{\frac{\alpha+\sqrt{\alpha^2+2}}{2} e^{-\frac{1}{2}\left(\alpha^2+\left(\frac{\alpha+\sqrt{\alpha^2+2}}{2}\right)^2\right)}I_0\left(\alpha\frac{\alpha+\sqrt{\alpha^2+2}}{2}\right)}\Bigg),
\end{align}
\begin{align}\label{eq_c2}
   c_2(\alpha) = \frac{\alpha+\sqrt{\alpha^2+2}}{2}+
   \frac{Q_1\left(\alpha,\frac{\alpha+\sqrt{\alpha^2+2}}{2}\right)}{\frac{\alpha+\sqrt{\alpha^2+2}}{2} e^{-\frac{1}{2}\left(\alpha^2+\left(\frac{\alpha+\sqrt{\alpha^2+2}}{2}\right)^2\right)}I_0\left(\alpha\frac{\alpha+\sqrt{\alpha^2+2}}{2}\right)}.
\end{align}
\end{lem}
\begin{proof}
See  \cite[Sec. II]{guo2020semilinear}.
\end{proof}

Moreover, we can make some second level approximations considering different ranges of $\alpha$ to further simplify notations. For more details, refer to \cite{guo2020semilinear}.
\begin{figure}
    \centering
    \includegraphics[width=0.7\columnwidth]{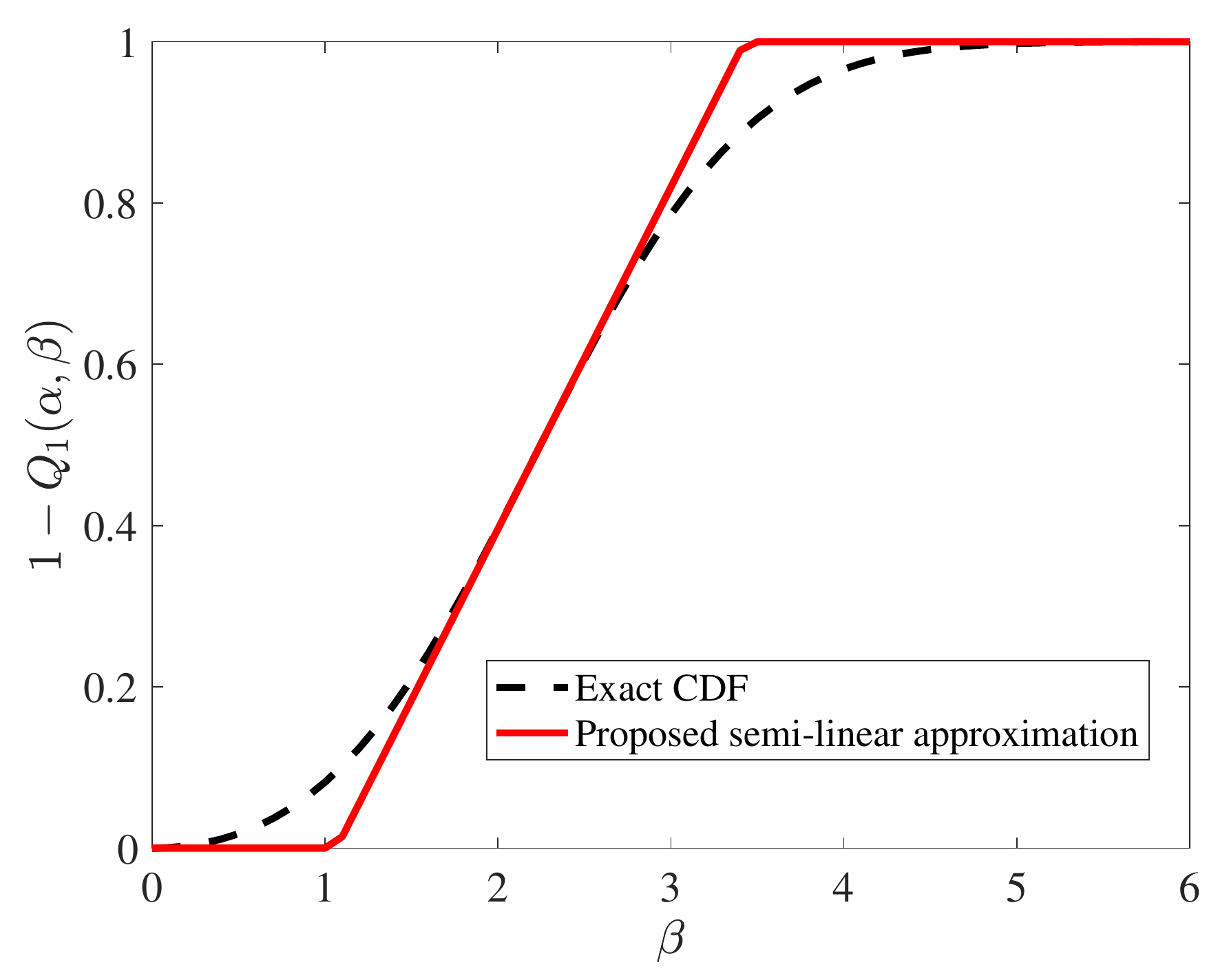}
    \caption{The illustration of proposed semi-linear approximation, $\alpha = 2$.}
    \label{fig:cdf}
\end{figure}

One example result of the proposed approximation can be seen in Fig. \ref{fig:cdf} with $\alpha$ set to 2. We can observe that Lemma 1 is tight for moderate values of $\beta$. Note that the proposed approximations are not tight at the tails of the \gls{cdf}. However, as observed in \cite{Bocus2013CLapproximation,helstrom1994elements,Makki2013TCfeedback,Makki2011Eurasipcapacity,Makki2018WCLwireless,Makki2016TVTperformance,Simon2003TWCsome,Suraweera2010TVTcapacity,Kang2003JSAClargest,Chen2004TCdistribution,Ma2000JSACunified,Zhang2002TCgeneral,Ghasemi2008ICMspectrum,Digham2007TCenergy, simon2002bookdigital,Fu2011GLOBECOMexponential,zhao2008ELtight,Simon2000TCexponential,Cao2016CLsolutions,sofotasios2015solutions,Cui2012ELtwo,annamalai2001WCMCcauchy,Sofotasios2010ISWCSnovel,Li2010TCnew,andras2011Mathematicageneralized,Gaur2003TVTsome,Kam2008TCcomputing,Corazza2002TITnew,Baricz2009TITnew,chiani1999ELintegral}, in different applications, the  Marcum $Q$-function is typically combined with other functions which tend to zero in the tails of the CDF. In such cases, the inaccuracy of the approximation at the tails does not affect the tightness of the final analysis. For example, it can simplify integrals such as
\begin{align}\label{eq_integral}
G(\alpha,\rho)=\int_\rho^\infty{e^{-nx} x^m \left(1-Q_1(\alpha,x)\right)\text{d}x} ~~\forall n,m,\alpha,\rho>0.
\end{align}
Such an integral has been observed in various applications, e.g., \cite[Eq. 1]{Simon2003TWCsome}, \cite[Eq. 2]{Cao2016CLsolutions}, \cite[Eq. 1]{sofotasios2015solutions}, \cite[Eq. 3]{Cui2012ELtwo}, and \cite[Eq. 1]{Gaur2003TVTsome}. However, depending on the values of $n, m$ and $\rho$, (\ref{eq_integral}) may have no closed-form expression.

Another example of integral calculation is
\begin{align}\label{eq_integralT}
    T(\alpha,m,a,\theta_1,\theta_2) = \int_{\theta_1}^{\theta_2} e^{-mx}\log(1+ax)Q_1(\alpha,x)\text{d}x ~~\forall m>0,a,\alpha,
\end{align}
with $\theta_2>\theta_1\geq0$, which does not have a closed-form expression for different values of $m, a, \alpha$. This integral is interesting as it is often used to analyse the expected performance of outage-limited systems, e.g,  \cite{Simon2003TWCsome,Simon2000TCexponential,Gaur2003TVTsome,6911973}.

Finally, the proposed semi-linear approximation can be used for the rate adaptation scheme developed in the \gls{pa} system. For more details,  refer to Chapter \ref{chapter:nosections} as well as \cite[Sec. II]{guo2020semilinear}.

%% file: YourThesis/chapters/Three.tex
\chapter{Resource Allocation in PA Systems}\label{chapter:nosections}
Resource allocation plays an important role in communication systems as a way of optimizing the assignment of available resources to achieve  network design objectives. In the \gls{pa} system, resource allocation can be deployed to mitigate different challenges mentioned in Chapter \ref{chapter:two}. In this chapter, we develop various resource allocation schemes for the \gls{pa} system under different practical constraints.

\section{Rate Adaptation in the Classic PA Setup}
In this section, we propose a rate adaptation scheme to mitigate the mismatch problem. Here, the classic setup means the \gls{pa} is only used for channel prediction, not for data transmission. We assume that $d_\text{a}$, $\delta $ and $\hat{g}$ are known at the BS. It can be seen from (\ref{eq_pdf}) that $f_{g|\hat{H}}(x)$ is a function of $v$. For a given $v$, the distribution of $g$ is known at the BS, and a rate adaption scheme can be performed to improve the system performance.

The data is transmitted with rate $R^*$ nats-per-channel-use (npcu). If the instantaneous realization of the channel gain supports the data rate, i.e., $\log(1+gP)\ge R^*$, the data can be successfully decoded, otherwise outage occurs.  Hence, the outage probability in each time slot is obtained as $\text{Pr}(\text{Outage}|\hat{H}) = F_{g|\hat{H}}\left(\frac{e^{R^*}-1}{P}\right)$. Considering slotted communication in block fading channels, where $\Pr(\text{Outage})>0$ varies with different fading models. Here, we define throughput as the data rate times the successful decoding probability \cite[p. 2631]{Biglieri1998TITfading}, \cite[Th. 6]{Verdu1994TITgeneral}, \cite[Eq. 9]{Makki2014TCperformance}. I.e., the expected data rate successfully received by the receiver is an appropriate performance metric. Hence, the rate adaptation problem of maximizing the  throughput in each time slot, with given $v$ and $\hat{g}$, can be expressed as
\begin{align}\label{eq_avgR}
    R_{\text{opt}|\hat{g}}=\argmax_{R^*\geq 0} \left\{ \left(1-\text{Pr}\left(\log(1+gP)<R^*\right)\right)R^* \right\}\nonumber\\
   =\argmax_{R^*\geq 0} \left\{ \left(1-\text{Pr}\left(g<\frac{e^{R^*}-1}{P}\right)\right)R^*\right \}\nonumber\\
    =\argmax_{R^*\geq 0} \left\{ \left(1-F_{g|\hat{H}}\left(\frac{e^{R^*}-1}{P}\right)\right) R^*\right\}.
\end{align}
Also, the expected throughput is obtained by $\mathbb{E}\left\{ \left(1-F_{g|\hat{H}}\left(\frac{e^{R_{\text{opt}|\hat g}}-1}{P}\right)\right) R_{\text{opt}|\hat g}\right\}$ with expectation over $\hat g$. 

Using (\ref{eq_cdf}), (\ref{eq_avgR}) is simplified as
\begin{eqnarray}\label{eq_optR}
    R_{\text{opt}|\hat{g}}=\argmax_{R^*\geq 0} \left\{ Q_1\left(  \sqrt{\frac{2\hat{g}(1-\sigma^2)}{\sigma^2}}, \sqrt{\frac{2(e^{R^*}-1)}{P\sigma^2}} \right)R^* \right\}.
\end{eqnarray}

In general, (\ref{eq_optR}) does not have a closed-form solution. For this reason, in \cite{guo2020semilinear} and \cite{Guo2019WCLrate} we propose different approximations for the optimal data rate maximizing the instantaneous throughput.

\section{Hybrid Automatic Repeat Request in the PA Systems}

\gls{harq} is a well-known approach to improve  data transmission reliability and efficiency. The main idea of \gls{harq} is to retransmit the message that experienced poor channel conditions in order to reduce the outage probability \cite{Caire2001throughput,makki2012hybrid,Makki2014TCperformance}. Here, we define that outage occurs when the transmitted message cannot be decoded at the receiver. For the Rayleigh block fading channel, infinite power is required to achieve zero outage probability for all realizations. Hence, we replace the strict outage constraint by a more realistic requirement, where a transmission is successful as long as the message can always be decoded by the receiver with probability $\epsilon$. We define $\epsilon$ as a parameter of the system outage tolerance.

The outage-constrained power allocation problem in \gls{harq} systems has been studied in, e.g, \cite{caire1999optimum} with perfect \gls{csi} assumption, in \cite{lee2010cooperative} with cooperative decode-and-forward \gls{arq} relaying under packet-rate fading channels, and in \cite{makki2014green,makki2014greenglobe} with power allocation schemes aiming at minimizing the average transmit power for \gls{harq} systems.  Also, in the  block fading scenario, \cite{wu2010performance} studied the outage-limited performance of different \gls{harq} protocols.  Moreover, the outage-constrained power optimization for the \gls{rtd} and fixed-length coding \gls{inr} \gls{harq} protocols were investigated in \cite{su2011optimal} and \cite{chaitanya2011outage}, respectively. Assuming that the channel changes in each re-transmission, \cite{seo2008optimal} developed power allocation schemes with basic \gls{arq}. Finally, \cite{djonin2008joint} proposed a linear-programming approach with a buffer cost constraint to solve the power adaptive problem in HARQ systems where the power is adapted based on the received \gls{csi}.

As the PA system includes the feedback link with the \gls{fdd} setup, i.e., from the \gls{pa} to the \gls{bs}, \gls{harq} can be supported by the \gls{pa} structure in high mobility scenarios. That is, the \gls{bs} could potentially adjust its transmit rate/power based on the feedback from the \gls{pa}. In this way, it is expected that the joint implementation of the \gls{pa} system and the \gls{harq} scheme can improve the efficiency and reliability of outage-constrained systems. There is limited work on deploying \gls{harq} in high mobility scenarios, i.e., when the channel change quickly over time compared to the feedback delay. In \cite{makki2012arq}, the authors investigated the performance of basic \gls{arq} and \gls{inr} protocols in fast-fading channels where a number of channel realizations are experienced in each retransmission round. Also, \cite{chelli2014performance}  studied the performance of \gls{inr} over double Rayleigh channels, a common model for the fading amplitude of vehicle-to-vehicle communication systems. However, both \cite{makki2012arq} and \cite{chelli2014performance} deal with the same channel \gls{pdf} for different re-transmission rounds, which has limited contribution for the spatial/temporal variation of the channel in vehicle communications.

In the classic \gls{pa} setup,  the spectrum is underutilized, and the spectral efficiency could be further improved in the case that the \gls{pa} could be used not only  for channel prediction, but also for data transmission. We address these challenges by implementing \gls{harq}-based protocols in \gls{pa} systems as follows.

\begin{figure}
    \centering
    \includegraphics[width=0.7\columnwidth]{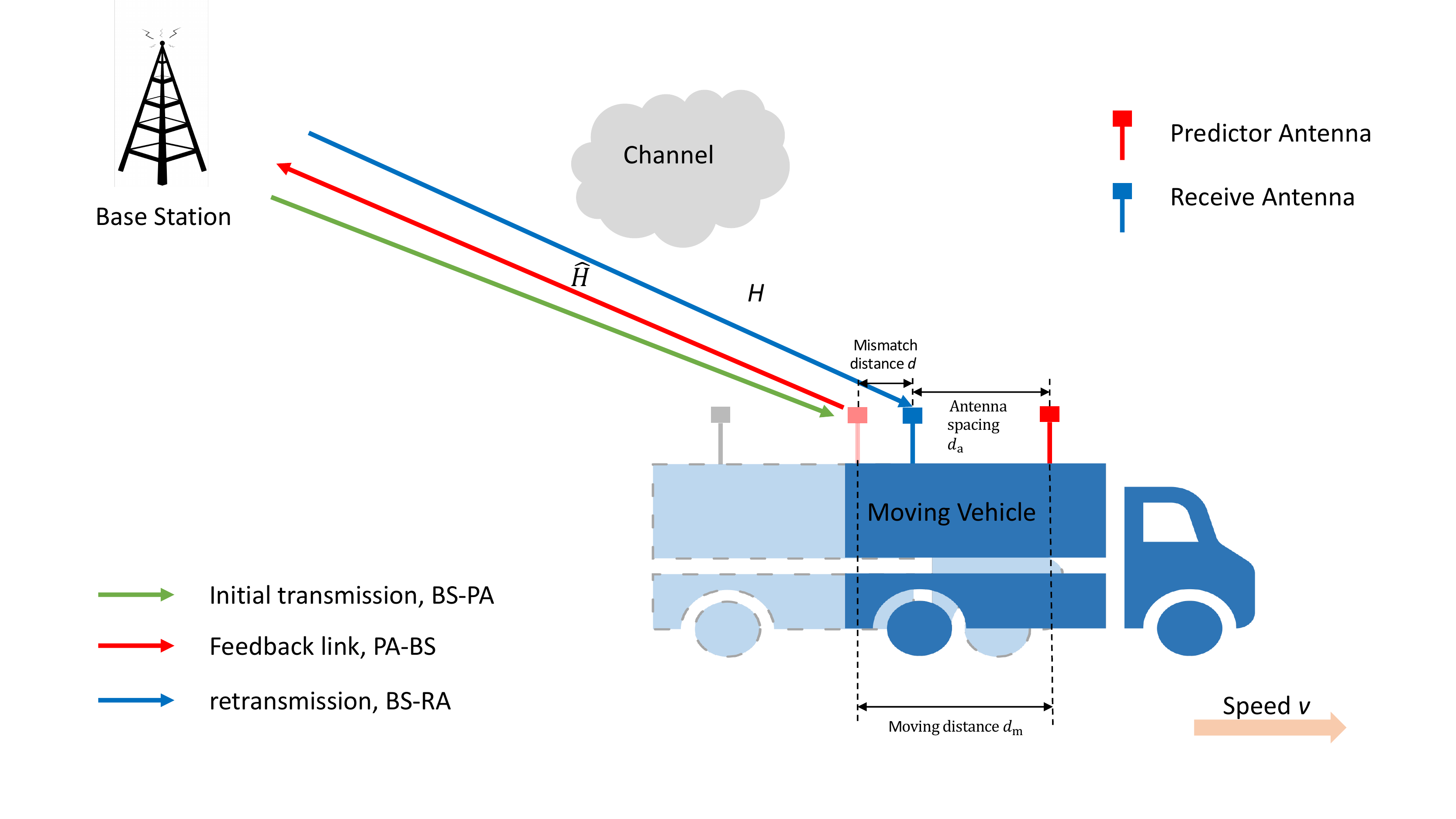}
    \caption{The PA concept with \gls{harq}.}
    \label{fig:PA_harq}
\end{figure}

As seen in Fig. \ref{fig:PA_harq}, the \gls{bs} first sends pilots as well as the encoded data with certain initial parameters, e.g., power or rate depending on the problem formulation, to the \gls{pa}. Then, the \gls{pa} estimates the channel from the received pilots. At the same time, the \gls{pa} tries to decode the signal. If the message is correctly decoded,  an \gls{ack} is fed back to the \gls{bs}, and the data transmission stops. Otherwise, the \gls{pa} sends both a \gls{nack} and high accuracy quantized \gls{csi} feedback. The number of quantization bits are large enough such that we can assume the BS to have perfect \gls{csi}. With \gls{nack}, in the second transmission round, the \gls{bs}  transmits the message to the \gls{ra} with adapted power/rate which is a function of the instantaneous channel quality in the first round. The outage occurs if the RA cannot decode the message at the end of the second round. 

In the following section, we develop an outage-constrained power allocation scheme in the \gls{harq}-based \gls{pa} system. The related study about delay-limited average rate optimization in the \gls{harq}-based \gls{pa} system can be found in \cite{guo2020rate}.

\section{ Outage-constrained Power Allocation in the HARQ-based PA Systems}
As  seen in Fig. \ref{fig:time_power}, with no \gls{csi}, at $t_1$ the \gls{bs} sends pilots as well as the encoded data with certain initial rate $R$ and power $P_1$ to the \gls{pa}. At $t_2$, the \gls{pa} estimates the channel $\hat{H}$ from the received pilots. At the same time, the \gls{pa} tries to decode the signal. If the message is correctly decoded, i.e., $R\leq \log(1+\hat{g}P_1)$,  an \gls{ack} is fed back to the \gls{bs} at $t_3$, and the data transmission stops. Otherwise, the \gls{pa} sends both a \gls{nack} and high accuracy quantized \gls{csi} feedback about $\hat{H}$.  With \gls{nack}, in the second transmission round at time $t_4$, the \gls{bs}  transmits the message to the \gls{ra} with power $P_2$,  which is a function of the instantaneous channel quality $\hat{g}$. The outage occurs if the \gls{ra} cannot decode the message at the end of the second round. 

\begin{figure}
    \centering
    \includegraphics[width=0.7\columnwidth]{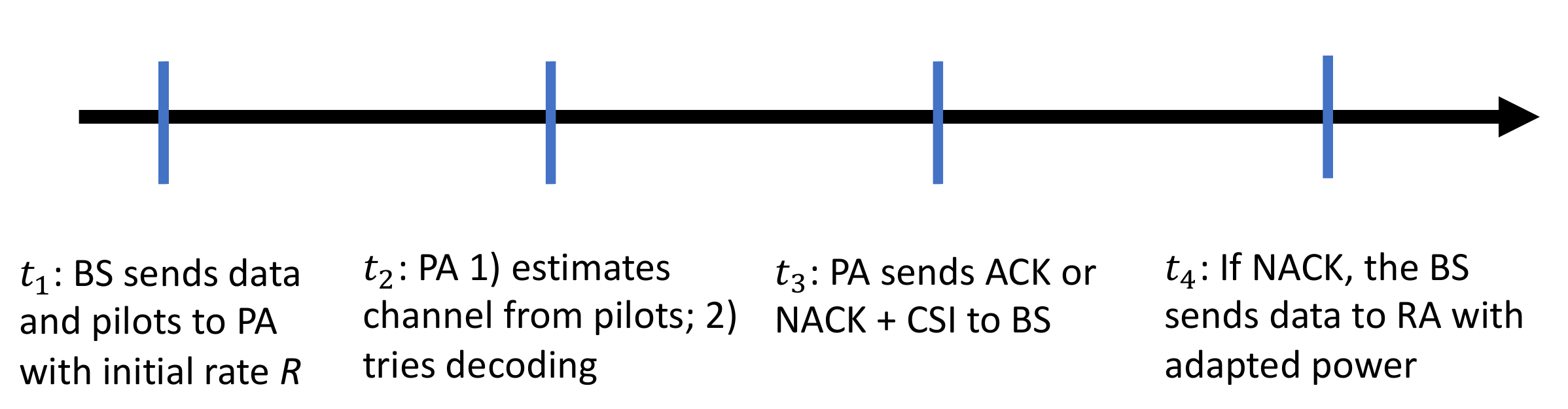}
    \caption{Time structure for the proposed outage-constrained power allocation in the \gls{harq}-based \gls{pa} system.}
    \label{fig:time_power}
\end{figure}

Let $\epsilon$ be the outage probability constraint. Here, we present the results for the cases with \gls{rtd} and \gls{inr} \gls{harq} protocols. With an \gls{rtd} protocol, the same signal (with possibly different power) is sent in each retransmission round, and the receiver performs maximum ratio combining of all received copies of the signal. With \gls{inr}, on the other hand, new redundancy bits are sent in the retransmissions, and the receiver decodes the message by combining all received signals \cite{makki2014green,chaitanya2011outage,djonin2008joint}.

Considering Rayleigh fading conditions with $f_{\hat{g}}(x) =  e^{- x}$, the outage probability at the end of Round 1 is given by
\begin{align}\label{eq_pout}
\text{Pr}(\text{Outage, Round 1}) = \text{Pr}\left\{R\leq \log(1+\hat{g}P_1)\right\} = \text{Pr}\left\{\hat{g}\leq \frac{e^{R}-1}{P_1}\right\}
 = 1-e^{ -\frac{\theta}{P_1}}, 
\end{align}
where $\theta = e^{R}-1$. Then, using the results of, e.g.,  \cite[Eq. 7, 18]{makki2014green} on the outage probability of the \gls{rtd}- and \gls{inr}-based \gls{harq} protocols, the power allocation problem for the proposed \gls{harq}-based \gls{pa} system can be stated as
\begin{equation}
\label{eq_optproblem}
\begin{aligned}
\min_{P_1,P_2} \quad & \mathbb{E}_{\hat{g}}\left[P_\text{tot}|\hat{g}\right] \\
\textrm{s.t.} \quad &  P_1, P_2 > 0,\\
&P_\text{tot}|\hat{g} = \left[P_1 + P_2(\hat{g}) \times \mathcal{I}\left\{\hat{g} \le \frac{\theta}{P_1}\right\}\right],
\end{aligned}
\end{equation}
with
\begin{align}\label{eq_optproblemrtd}
F_{g|\hat{g}}\left\{\frac{\theta-\hat{g}P_1}{P_2(\hat{g})}    \right\} = \epsilon, \quad\text{for RTD}
\end{align}
\begin{align}\label{eq_optprobleminr}
F_{g|\hat{g}}\left\{\frac{e^{R-\log(1+\hat{g}P_1)}-1}{P_2(\hat{g})}    \right\} = \epsilon, \quad\text{for INR}. 
\end{align}
Here, $P_\text{tot}|\hat{g}$ is the total instantaneous transmission power for two transmission rounds (i.e., one retransmission) with given $\hat{g}$, and we define $\Bar{P} \doteq \mathbb{E}_{\hat{g}}\left[P_\text{tot}|\hat{g}\right]$ as the expected power, averaged over $\hat{g}$. Moreover, $\mathcal{I}(x)=1$ if $x>0$ and $\mathcal{I}(x)=0$ if $x \le 0$. Also, $\mathbb{E}_{\hat{g}}[\cdot]$ represents the expectation operator over $\hat{g}$. Here, we ignore the peak power constraint and assume that the \gls{bs} is capable for  sufficiently high transmission powers. Finally, (\ref{eq_optproblem})-(\ref{eq_optprobleminr}) come from the fact that, with our proposed scheme, $P_1$ is fixed and optimized with no \gls{csi} at the \gls{bs} and based on average system performance. On the other hand, $P_2$ is adapted continuously based on the instantaneous \gls{csi}.

Using (\ref{eq_optproblem}), the required power in Round 2 is given by
\begin{equation}
\label{eq_PRTDe}
    P_2(\hat{g}) = \frac{\theta-\hat{g}P_1}{F_{g|\hat{g}}^{-1}(\epsilon)},
\end{equation}
for the \gls{rtd}, and
\begin{equation}
\label{eq_PINRe}
    P_2(\hat{g}) = \frac{e^{R-\log(1+\hat{g}P_1)}-1}{F_{g|\hat{g}}^{-1}(\epsilon)},
\end{equation}
for the \gls{inr}, where $F_{g|\hat{g}}^{-1}(\cdot)$ is the inverse of the \gls{cdf} given in (\ref{eq_cdf}). Note that $F_{g|\hat{g}}^{-1}(\cdot)$ is a complicated function of $\hat{g}$,  and consequently, it is not possible to express $P_2$ in closed-form. For this reason, one can use   \cite[Eq. 2, 7]{6414576}
\begin{align}
    Q_1 (s, \rho) &\simeq e^{\left(-e^{\mathcal{I}(s)}\rho^{\mathcal{J}(s)}\right)}, \nonumber\\
    \mathcal{I}(s)& = -0.840+0.327s-0.740s^2+0.083s^3-0.004s^4,\nonumber\\
    \mathcal{J}(s)& = 2.174-0.592s+0.593s^2-0.092s^3+0.005s^4,
\end{align}
to approximate $F_{g|\hat{g}}$ and consequently $F_{g|\hat{g}}^{-1}(\epsilon)$. In this way, (\ref{eq_PRTDe}) and (\ref{eq_PINRe}) can be approximated as
\begin{equation}
\label{eq_PRTDa}
    P_2(\hat{g}) = \Omega\left(\theta-\hat{g}P_1\right),
\end{equation}
for the RTD, and
\begin{equation}
\label{eq_PINRa}
   P_2(\hat{g}) = \Omega\left(e^{R-\log(1+\hat{g}P_1)}-1\right),
\end{equation}
for the INR, where
\begin{equation}\label{eq_omega}
    \Omega (\hat{g}) = \frac{2}{\sigma^2}\left(\frac{\log(1-\epsilon)}{-e^{\mathcal{I}\left(\sqrt{\frac{2\hat{g}(1-\sigma^2)}{\sigma^2}}\right)}}\right)^{-\frac{2}{\mathcal{J}\left(\sqrt{\frac{2\hat{g}(1-\sigma^2)}{\sigma^2}}\right)}}.
\end{equation}

In this way, for different \gls{harq} protocols, we can express the  instantaneous transmission power of Round 2, for every given $\hat{g}$ in closed-form. Then, the  power allocation problem (\ref{eq_optproblem}) can be solved numerically. However, (\ref{eq_omega}) is still complicated and it is not possible to solve (\ref{eq_optproblem}) in closed-form. For this reason, we propose an approximation scheme to solve (\ref{eq_optproblem}) as follows.

Let us initially concentrate on the \gls{rtd} protocol. Then,  combining (\ref{eq_optproblem}) and (\ref{eq_PRTDe}), the expected total transmission power  is given by
\begin{align}\label{eq_barP}
    \Bar{P}_\text{RTD} = P_1 + \int_0^{\theta/P_1} e^{- x}P_2\text{d}x
    = P_1 + \int_0^{\theta/P_1} e^{- x}\frac{\theta-x P_1}{F_{g|x}^{-1}(\epsilon)}\text{d}x.
\end{align}

Then, Theorem 1 in \cite{guo2020power} derives the minimum required power  in Round 1 and the average total power consumption.

To study the performance of the \gls{inr}, we can use Jensen's inequality and the concavity of the logarithm function \cite[Eq. 30]{makki2016TWCRFFSO}
\begin{align}\label{eq_jensen}
   \frac{1}{n}\sum_{i=1}^{n} \log (1+x_i)\leq\log\left(1+\frac{1}{n}\sum_{i=1}^{n}x_i\right),
\end{align}
and derive the closed-form expressions for the minimum required power following the similar steps as for the \gls{rtd} (see \cite[Sec. III B]{guo2020power}) for detailed discussions).

%% file: YourThesis/chapters/Conclusion.tex
\chapter{Contributions and Future Work}\label{sec:conclusion}
 This chapter summarizes the contributions of each appended publication and lays out possible directions for future work based on the topics in this thesis.
\section{Paper A}
\subsection*{"Rate adaptation in predictor antenna systems"}
In this paper, we study the performance of \gls{pa} systems in the presence of the mismatch problem with rate adaptation. We derive closed-form expressions for the instantaneous throughput, the outage probability, and the throughput-optimized rate adaptation. Also, we take the temporal variation of the channel into account and evaluate the system performance in various conditions. The simulation and the analytical results show that, while \gls{pa}-assisted adaptive rate adaptation leads to considerable performance improvement, the throughput and the outage probability are remarkably affected by the spatial mismatch and temporal correlations.

\section{Paper B}
\subsection*{"A semi-linear approximation of the first-order Marcum $Q$-function with application to predictor antenna systems"}
In this paper, we first present a semi-linear approximation of the Marcum $Q$-function. Our proposed approximation is useful because it simplifies, e.g., various integral calculations including the Marcum $Q$-function as well as various operations such as parameter optimization. Then, as an example of interest, we apply our proposed approximation approach to the performance analysis of \gls{pa} systems. Considering spatial mismatch due to  mobility, we derive closed-form expressions for the instantaneous and average throughput as well as the throughput-optimized rate allocation. As we show, our proposed approximation scheme enables us to analyze \gls{pa} systems with high accuracy. Moreover, our results show that rate adaptation can improve the performance of \gls{pa} systems with different levels of spatial mismatch.

\section{Paper C}
\subsection*{"Power allocation in HARQ-based predictor antenna systems"}
In this work, we study the performance of \gls{pa} systems using \gls{harq}.  Considering spatial mismatch due to the vehicle mobility, we derive closed-form expressions for the optimal power allocation and the minimum average power of the \gls{pa} systems under different outage probability constraints. The results are presented for different types of \gls{harq} protocols, and we study the effect of different parameters on the performance of \gls{pa} systems. As we show, our proposed approximation scheme enables us to analyze \gls{pa} systems with high accuracy. Moreover, for different vehicle speeds, we show that the \gls{harq}-based feedback can reduce the outage-limited transmit power consumption of PA systems by orders of magnitude.


\section{Related Contributions}
Another \gls{csi}-related application in vehicle communication is beamforming. As discussed in Chapter \ref{ch:introduction}, the channel for vehicles changes rapidly, such that it is hard to acquire \gls{csit}, especially during initial access. In Paper D, we study the performance of large-but-finite \gls{mimo}  networks using codebook-based beamforming. Results show that the proposed genetic algorithm-based scheme can reach (almost) the same performance as in the exhaustive search-based scheme with considerably lower implementation complexity. Then, in Paper E, we extend our study in Paper D to include beamforming at both the transmitter and the receiver side. Also, we compare different machine learning-based analog beamforming approaches for the beam refinement. As indicated in the  results, our scheme outperforms the considered state-of-the-art schemes in terms of throughput. Moreover, when taking the user mobility into account, the proposed approach can remarkably reduce the algorithm running delay based on the beamforming results in the previous time slots. Finally, in Paper F, with collaborative users, we show that the end-to-end throughput can be improved by  data exchange through device-to-device links among the users.

\section{Future work}
In this thesis, we have developed analytical models for evaluating the \gls{pa} system from  different perspectives, and proposed resource allocation schemes to mitigate the mismatch problem and improve the system performance. Here are some potential directions for future work:
\begin{itemize}
    \item Several results presented in the papers above rely on the assumption that the scattering environment around the \glspl{ra} is isotropic and remains constant over the time period of interest in a small moving distance.  To more accurately resemble reality, one could consider alternative models to evaluate the \gls{pa} system, for example some mixture models with more time-varying properties \cite{abdi2003new}.
    \item As a natural follow-up from above, one can consider more use cases for the \gls{pa} system, such as satellite-train communication and vehicle localization, with different channel models and service requirements. Here, the results of  \cite{shim2016traffic,lannoo2007radio,wang2012distributed,Dat2015,Laiyemo2017,yutao2013moving,Marzuki2017,Andreev2019,Haider2016,Patra2017} can be supportive. 
    \item The work we have done considers \gls{siso} setup, i.e., with one antenna at the \gls{bs} and one \gls{ra} on the top of the vehicle at the receiver side. Though in \cite{guo2020power}  we exploit the \gls{pa} as part of the data transmission, it is still interesting to see where  the gain of deploying the \gls{pa} system comes from and when we should apply it over typical transceiver schemes. Moreover, one can deploy the \gls{pa}  in multiple antenna systems for which the results of  \cite{Dinh2013ICCVEadaptive,DT2015ITSMmaking,phan2018WSAadaptive} can be useful. It is expected that combining the \gls{pa} with \gls{mimo} would result in higher performance gain in fast moving scenarios. 
    \item As we discussed in Chapter 1, in \gls{pa} systems we target at \gls{urllc}, i.e., delay/latency is crucial in the system design. Hence, there is a natural extension to perform finite blocklength analysis in the (\gls{harq}-based) \gls{pa} system. As opposed to the literature on finite blocklength studies, e.g., \cite{makki2014greenglobe,Makki2014WCLfinite,Yang2014TITquasi}, here the channel in the retransmission round(s) is different from the one in the initial transmission due to mobility.
    \item Machine learning-based channel estimation/prediction has become powerful in various applications where the statistical model of the channel does not exist or is not robust \cite{ye2017power,wen2018deep,feng2020deep}. On the other hand, the \gls{pa} itself provides a reliable feedback loop at the cost of additional resources. Using the \gls{pa} setup to perform machine learning-based channel prediction would be a very valuable contribution.
\end{itemize}